\documentclass[a4paper,UKenglish,cleveref, autoref, thm-restate]{lipics-v2021}

  \usepackage{bbm}
\title{The Fine-Grained Complexity of Boolean Conjunctive Queries and Sum-Product Problems}

\titlerunning{The Fine-Grained Complexity of BCQs and Sum-Product Problems}

\author{Austen Z. Fan}{Department of Computer Sciences, University of Wisconsin-Madison, United States}{afan@cs.wisc.edu}{https://orcid.org/0000-0001-7714-2195}{}
\author{Paraschos Koutris}{Department of Computer Sciences, University of Wisconsin-Madison, United States}{paris@cs.wisc.edu}{https://orcid.org/0000-0001-6309-1702}{}

\author{Hangdong Zhao}{Department of Computer Sciences, University of Wisconsin-Madison, United States}{hangdong@cs.wisc.edu}{https://orcid.org/0009-0009-7636-0831}{}

\authorrunning{A.\,Z. Fan, P. Koutris and H. Zhao} %TODO mandatory. First: Use abbreviated first/middle names. Second (only in severe cases): Use first author plus 'et al.'

\Copyright{Austen Z. Fan, Paraschos Koutris and Hangdong Zhao} %TODO mandatory, please use full first names. LIPIcs license is "CC-BY";  http://creativecommons.org/licenses/by/3.0/

 %TODO mandatory; please add comma-separated list of keywords

\category{Track B: Automata, Logic, Semantics, and Theory of Programming} %optional, e.g. invited paper
\relatedversion{} %optional, e.g. full version hosted on arXiv, HAL, or other respository/website
\relatedversiondetails[cite=arxiv]{Full Version}{http://arxiv.org/abs/2304.14557} %linktext and cite are optional

\funding{NSF IIS-1910014}
\nolinenumbers %uncomment to disable line numbering

\EventEditors{Kousha Etessami, Uriel Feige, and Gabriele Puppis}
\EventNoEds{3}
\EventLongTitle{50th International Colloquium on Automata, Languages, and Programming (ICALP 2023)}
\EventShortTitle{ICALP 2023}
\EventAcronym{ICALP}
\EventYear{2023}
\EventDate{July 10--14, 2023}
\EventLocation{Paderborn, Germany}
\EventLogo{}
\SeriesVolume{261}
\ArticleNo{128}
\usepackage[linesnumbered,ruled,noend]{algorithm2e}

\usepackage{tikz}
\usetikzlibrary{positioning,calc}
\usetikzlibrary{decorations.text}
\usetikzlibrary{decorations.pathmorphing,decorations.pathreplacing}
\usetikzlibrary{arrows,petri, topaths,fit}

\usepackage{listings}
\usepackage{wrapfig}
\usepackage{framed}   % for framing
 {\endMakeFramed}
\definecolor{shadecolor}{gray}{0.75}
\usepackage{thmtools, thm-restate}

\usepackage{multirow}
\usepackage{subcaption}
\usepackage{url}
\usepackage{graphicx}
\usepackage{tcolorbox}% http://ctan.org/pkg/tcolorbox
\usepackage{enumitem}

\definecolor{mycolor}{rgb}{0.122, 0.435, 0.698}% Rule colour
\makeatletter
\newcommand{\mybox}[1]{%
  \setbox0=\hbox{#1}%
  \setlength{\@tempdima}{\dimexpr\wd0+13pt}%
  \begin{tcolorbox}[colframe=mycolor,boxrule=0.5pt,arc=4pt,
      left=6pt,right=6pt,top=3pt,bottom=3pt,boxsep=0pt,width=\@tempdima]
    #1
  \end{tcolorbox}
}

\newcommand{\emb}[0]{\mathsf{emb}}

\newcommand{\mH}{\mathcal{H}}

\newcommand{\bx}{\mathbf{x}}
\newcommand{\ba}{\mathbf{a}}

\newcommand{\fhw}{\mathsf{fhw}}
\newcommand{\subw}{\mathsf{subw}}
\newcommand{\adw}{\mathsf{adw}}

\newcommand{\PANDA}{\mathsf{PANDA}}

\newcommand{\bags}{\mathsf{bags}}

\newcommand{\htree}{\mathcal{T}}
\newcommand*{\defeq}{\stackrel{\text{def}}{=}}

\begin{document}

\maketitle

% \author{Paraschos Koutris}
% \affiliation{%
% \institution{University of Wisconsin-Madison}
% \city{Madison}
% \state{WI}
% \country{USA}}

% \author{Shaleen Deep}
% \affiliation{%
% \institution{University of Wisconsin-Madison}
% \city{Madison}
% \state{WI}
% \country{USA}}

\begin{abstract}
We study the fine-grained complexity of evaluating Boolean Conjunctive Queries and their generalization to sum-of-product problems over an arbitrary semiring. For these problems, we present  a general \emph{semiring-oblivious} reduction from the $k$-clique problem to any query structure (hypergraph). Our reduction uses the notion of {\em embedding} a graph to a hypergraph, first introduced by Marx~\cite{Marx13}. As a consequence of our reduction, we can show tight conditional lower bounds for many classes of hypergraphs, including cycles, Loomis-Whitney joins, some bipartite graphs, and chordal graphs. These lower bounds have a dependence on what we call the {\em clique embedding power} of a hypergraph $H$, which we believe is a quantity of independent interest. We show that the clique embedding power is always less than the submodular width of the hypergraph, and present a decidable algorithm for computing it. We conclude with many open problems for future research.
\end{abstract}

\section{Introduction}

In a seminal paper, Marx proved the celebrated result that $\operatorname{CSP}(\mathcal{H})$ is fixed-parameter tractable (FPT) if and only if the hypergraph $\mH$ has a bounded submodular width~\cite{Marx13}. In the language of database theory, a Boolean Conjunctive Query (BCQ) can be identified as the problem of $\operatorname{CSP}(\mH)$ where $\mH$ is the hypergraph associated with the query~\cite{grohe2007complexity}. Thus, Marx's result implies that a class of Boolean Conjunctive Queries is FPT if and only if its submodular width is bounded above by some universal constant. Built on this result, Khamis, Ngo, and Suciu introduced in~\cite{PANDA} the $\mathsf{PANDA}$ (Proof-Assisted eNtropic Degree-Aware) algorithm, which can evaluate a BCQ\footnote{Technically, the $\PANDA$ algorithm works for Boolean or full CQs.} in time $\tilde{O}(|I|^{\mathsf{subw}(H)})$, where $|I|$ is the input size and $\mathsf{subw}(\mH)$ is the submodular width of $\mH$ (here $\tilde{O}$ hides polylogarithmic factors). Remarkably, the running time of $\mathsf{PANDA}$ achieves the best known running time of \emph{combinatorial algorithm}\footnote{Informally speaking, this requires the algorithm does not leverage fast matrix multiplication techniques} for \emph{all} BCQs. It is thus an important open question whether there exists a faster combinatorial algorithm than $\mathsf{PANDA}$ for some Boolean CQ.

To show that large submodular width implies not being FPT, Marx introduced the notion of an \emph{embedding}, which essentially describes a reduction from one $\operatorname{CSP}$ problem to another. Our key insight in this work is that we can apply the notion of an embedding to measure how well {\em cliques} of different sizes can be embedded to a hypergraph $\mH$. By taking  the supremum over all possible clique sizes, we arrive at the definition of \emph{clique embedding power}, denoted $\emb(\mH)$. The use of cliques as the starting problem means that we can use popular lower bound conjectures in fine-grained complexity (the Boolean $k$-Clique conjecture, the Min-Weight $k$-Clique conjecture) to obtain (conditional) lower bounds for the evaluation of BCQs that depend on $\emb(\mH)$.

Equipped with the new notion of the clique embedding power, we can show tight lower bounds for several classes of queries. That is, assuming the Boolean $k$-Clique Conjecture, we derive (conditional) lower bounds for many queries that meet their submodular width, and therefore the current best algorithm, up to polylogarithmic factors. In particular, we show that for cycles~\cite{DBLP:journals/algorithmica/AlonYZ97}, Loomis-Whitney joins~\cite{NgoPRR18}, and chordal graphs, among others, the current combinatorial algorithms are optimal.

We further extend the embedding reduction to be independent of the underlying (commutative) semiring\footnote{A triple $(\mathbf{D}, \oplus, \otimes, \mathbf{0}, \mathbf{1})$ is a commutative semiring if $\oplus$ and $\otimes$ are commutative binary operators over $\mathbf{D}$ with the following properties:
$(i)$ $(\mathbf{D}, \oplus)$ is a commutative monoid with an additive identity $\mathbf{0}$. 
$(ii)$ $(\mathbf{D}, \otimes)$ is a commutative monoid with a multiplicative identity $\mathbf{1}$. 
$(iii)$ $\otimes$ distributes over $\oplus$. 
$(iv)$ For any element $e \in \mathbf{D}$, we have $e \otimes \mathbf{0} = \mathbf{0}  \otimes e = \mathbf{0} $.}. It was observed by Green, Karvounarakis, and Tannen~\cite{GreenKT07} that the semantics of CQs can be naturally generalized to sum-of-product operations over a semiring. This point of view unifies a number of database query semantics that seem unrelated. For example, evaluation over set semantics corresponds to evaluation over the Boolean semiring $\sigma_\mathbb{B}= (\{0,1\},\vee, \wedge,0,1)$, while bag semantics corresponds to the semiring $(\mathbb{N},+, \times,0,1)$. Interestingly, following this framework, the decision problem of finding a $k$-clique in a graph can be interpreted as the following sum-of-product operation: consider the input graph $G = (V,E)$ as the edge-weighted graph of the complete graph with $|V|$ vertices where $\mathsf{weight}(e) = \mathbbm{1}_{e\in E}$; then the problem is to compute $\bigvee_{{V' \subseteq V : |V'| = k}} \bigwedge \mathsf{w}(\{v,w\})$. Observe that by changing the underlying semiring to be the tropical semiring $\mathsf{trop} = (\mathbb{R}^\infty, \min, +, \infty,0)$, this formulation computes the min-weight $k$-clique problem. Indeed, given an edge-weighted graph (where the weight of non-existence edges is 0), the minimum weight of its $k$-clique is exactly $\min_{{V' \subseteq V: |V'| = k}} \sum \mathsf{w}(\{v,w\})$. We prove that the clique embedding reduction is \emph{semiring-oblivious}, i.e., the reduction holds for arbitrary underlying semirings. This enables one to transfer the lower bound result independent of the underlying semiring and should be of independent interest.

Recent years have witnessed emerging interests in proving lower bounds for the runtime of database queries (see Durand~\cite{survey} for a wonderful survey). Casel and Schmid consider the fine-grained complexity of regular path queries over graph databases~\cite{CaselS21}. 
Joglekar and R{\'{e}} prove a full dichotomy for whether a 1-series-parallel graph admits a subquadratic algorithm~\cite{JoglekarR18}. Their proof is based on the hardness hypothesis that 3-XOR  cannot be solved in subquadratic time. Perhaps the line of work in spirit closest to ours is the characterization of queries which can be enumerated by linear preprocessing time and constant delay~\cite{BaganDG07,CarmeliS22,Nofar22}. However, their results focus on the enumeration problem and therefore are different from the main subject of our paper. Furthermore, their characterization mainly classifies queries based on the existence of a linear preprocessing time and constant delay algorithm. In contrast, our method can provide a lower bound for \emph{every} query. 
%Based on hypotheses on Orthogonal Vectors, Boolean Matrix Multiplication, Sparse Boolean Matrix Multiplication and Online Matrix-Vector Multiplication, they essentially establish the optimality of the Product Graph approach. 

\subparagraph*{Our Contributions} We summarize our contributions as follows:
\begin{itemize}
\item We introduce the notion of the {\em clique embedding power} $\emb(\mH)$ of a hypergraph $\mH$  (Section~\ref{sec:emb}). We show several interesting properties of this notion; most importantly, we show that it is always upper-bounded by the submodular width, $\subw(\mH)$. This connection can be seen as additional evidence of the plausibility of the lower bound conjectures for the $k$-clique problem.
\item We show how to construct a reduction from the $k$-clique problem to any hypergraph $\mH$ for any semiring, and discuss how the clique embedding power provides a lower bound for its running time (Section~\ref{sec:lower:bounds}).
\item We study how to compute $\emb(\mH)$ (Section~\ref{sec: decidability}). In particular, we prove that it is a decidable problem, and give a Mixed Integer Linear Program formulation. One interesting consequence of this formulation is that to achieve $\emb(\mH)$ it suffices to consider clique sizes that depend on the hypergraph size.
\item We identify several classes of hypergraphs for which $\emb(\mH) = \subw(\mH)$ (Section~\ref{sec:tight}). For these classes of queries, our lower bounds match the best-known upper bounds if we consider the Boolean semiring with combinatorial algorithms or the tropical semiring. The most interesting class of hypergraphs we consider is the class of {\em chordal hypergraphs} (which captures chordal graphs).

\item Finally, we identify a hypergraph with six vertices for which there is a gap between its clique embedding power and submodular width (Section~\ref{sec:boat}). We believe that the existence of this gap leaves many open questions.
\end{itemize}
\section{Background}

In this section, we define the central problem, and notions necessary for our results.

\subparagraph*{The SumProduct Problem}
We define this general problem following the notation in~\cite{FAQ,FAQAI}.
Consider $\ell$ variables $x_1, x_2, \dots, x_\ell$, where each variable $x_i$ takes values in some discrete domain $\mathsf{Dom}(x_i)$. A {\em valuation} $v$ is a function that maps each $x_i$ to $\mathsf{Dom}(x_i)$. For a subset $S \subseteq [\ell]$, we define the tuple $\bx_S = (x_i)_{i \in S}$ and $v(\bx_S) = (v(x_i))_{i \in S}$.

The SumProduct Problem is parameterized by:
\begin{enumerate}
\item a commutative semiring $\sigma = (\mathbf{D}, \oplus, \otimes,\mathbf{0}, \mathbf{1})$, where $\mathbf{D}$ is a fixed domain.  
\item a hypergraph $\mH = (V,E)$ where $V = [\ell]$.
\end{enumerate}

The input $I$ specifies for every hyperedge $e \in E$ a function $R_e: \prod_{i \in e} \mathsf{Dom}(x_i) \rightarrow \mathbf{D}$. This function is represented in the input as a table of all tuples of the form $(\ba_e, R_e(\ba_e))$, such that  $R_e(\ba_e) \neq \mathbf{0}$. This input representation is standard in the CSP and database settings. We use $|I|$ to denote the input size, which is simply the sum of sizes of all tables in the input.

The SumProduct Problem then asks to compute the following function:
$$  \bigoplus_{v: \text{valuation}} \bigotimes_{e \in E} R_e(v(\bx_e)).$$
We will say that $v$ is a {\em solution} for the above problem if $\bigotimes_{e \in E} R_e(v(\bx_e)) \neq \mathbf{0}$.

Within this framework, we can capture several important problems depending on the choice of the semiring and the hypergraph. If we consider the Boolean semiring  $\sigma_\mathbb{B}= (\{0,1\},\vee, \wedge,0,1)$, then each $R_e$ behaves as a relational instance ($R_e$ is $1$ if the tuple is in the instance, otherwise $0$) and the SumProduct function captures Boolean Conjunctive Query evaluation. If $\sigma = (\mathbb{N}, +, \times,0,1)$ and $R_e$ is defined as above, then the SumProduct function computes the number of solutions to a Conjunctive Query. Another important class of problems is captured when we consider the min-tropical semiring $\mathsf{trop} = (\mathbb{R}^\infty, \min, +, \infty,0)$ and we assign each tuple to a non-negative weight; this computes a minimum weight solution that satisfies the structural properties. 

\subparagraph*{The Complexity for SumProduct Problems} We adopt the {\it random-access machine (RAM)} as our computation model with $O(\log n)$-bit
words, which is standard in fine-grained complexity. The machine has read-only input registers and it contains the database and the query, read-write work memory registers, and write-only output registers. It is assumed that each register can store any tuple, and each tuple is stored in one register. The machine can perform all ``standard''~\footnote{This includes all arithmetic (e.g. $+, -, \div, *$) and logical operations.} operations on one or two registers in constant time.

In this paper, we are interested in the computational complexity of a SumProduct problem for a fixed hypergraph $\mH$. (This is typically called {\em data complexity}). We will consider two different ways of treating semirings when we think about algorithms.

In the first variant, we fix the semiring $\sigma$ along with the hypergraph $\mH$. This means that the representation of the semiring is not part of the input and is known a priori to the algorithm. We denote this problem as $\mathsf{SumProd}\langle \sigma, \mH \rangle$.
In the second variant, we consider algorithms that access the semiring only via an oracle. In particular, the algorithm does not know the semiring a priori and can only access it during runtime by providing the values for the $\oplus, \otimes$ operations. We assume that each of these operations takes a constant amount of time. We denote this problem as $\mathsf{SumProd}\langle \mH \rangle$.

Our goal in this paper is to specify the exact exponent of $|I|$ in the polynomial-time runtime cost of an algorithm that computes  $\mathsf{SumProd}\langle \sigma, \mH \rangle$ or $\mathsf{SumProd}\langle \mH \rangle$.

\subparagraph*{Tree Decompositions}
A {\em tree decomposition} of a hypergraph $\mH$ is a pair $(\htree,\chi)$, where $\htree$ is a tree and $\chi$ maps each node $t \in V(\htree)$ of the tree to a subset $\chi(t)$ of $V(\mH)$ such that:
\begin{enumerate} 
\item every hyperedge $e \in E(\mH)$ is a subset of $\chi(t)$ for some  $t \in V(\htree)$; and 
\item for every vertex $v \in V(\mH)$, the set $\{t \mid v \in \chi(t) \}$ is a non-empty connected subtree of $\htree$. 
\end{enumerate}

We say that a hypergraph $\mH$ is \textit{acyclic} if it has a tree decomposition such that each bag corresponds to a hyperedge.

\subparagraph*{Notions of Width}
Let $\mH$ be a hypergraph and $F$ be a set function over $V(H)$. The $F$-width of a tree decomposition $(\htree,\chi)$ is defined as $\max_{t} F(\chi(t))$. The $F$-width of $\mH$ is the minimum $F$-width over all possible tree decompositions of $\mH$.

A {\em fractional independent set} of a hypergraph $\mH$ is a mapping $\mu : V(\mH) \rightarrow [0,1]$ such that $\sum_{v \in e} \mu(v) \leq 1$ for every $e \in E(\mH)$. We naturally extend functions on the vertices of $\mH$ to subsets of vertices of $\mH$ by setting $\mu(X) = \sum_{v \in X} \mu(v)$.

The {\em adaptive width} $\mathsf{adw}(\mH)$ of a hypergraph $\mH$ is defined as the supreme of $F\text{-width}(\mH)$, where $F$ goes over all fractional independent sets of $\mH$. Hence if $\mathsf{adw}(\mH) \leq w$, then for every $\mu$, there exists a tree decomposition of $\mH$ with $\mu$-width at most $w$. 

A set function $F$ is {\em submodular} if for any two sets $A,B$ we have $F(A \cup B) + F(A \cap B) \leq F(A) + F(B)$. It is monotone if whenever $A \subseteq B$, then $F(A) \leq F(B)$. The {\em submodular width} $\subw(\mH)$ of a hypergraph $\mH$ is defined as the supreme of $F\text{-width}(\mH)$, where $F$ now ranges over all non-negative, monotone, and submodular set functions over $V(\mH)$ such that for every hyperedge $e \in E(\mH)$, we have $F(e) \leq 1$. A non-negative, monotone, and submodular set function $F$ is \textit{edge-dominated} if $F(e) \leq 1$, for every $e \in E$.

The \textit{fractional hypertree width} of a hypergraph $\mH$ is $\fhw(\mH) = \min _{(\htree, \chi)} \max _{t \in V(\htree)} \rho^*(\chi(t))$, where $\rho^*$ is the minimum fractional edge cover number of the set $\chi(t)$. It  holds that $ \mathsf{adw}(\mH) \leq  \mathsf{subw}(\mH) \leq \fhw(\mH)$.

It is known that $\mathsf{SumProd} \langle \sigma_{\mathbb{B}}, \mH\rangle$ can be computed in time $\tilde{O}(|I|^{\mathsf{subw}(\mH)})$ using the $\PANDA$ algorithm~\cite{PANDA}. However, we do not know of a way to achieve the same runtime for the general $\mathsf{SumProd} \langle  \mH\rangle$ problem. For this, the best known runtime is  $\tilde{O}(|I|^{\mathsf{\#\subw}(\mH)})$, where $\subw(\mH) \leq \#\subw(\mH) \leq \fhw(\mH)$~\cite{FAQAI}.
On the other hand, there are hypergraphs for which we can compute  $\mathsf{SumProd} \langle \sigma_{\mathbb{B}}, \mH\rangle$ with runtime better than $\tilde{O}(|I|^{\mathsf{subw}(\mH)})$ using non-combinatorial algorithms. For example, if $\mH$ is a triangle we can obtain a runtime $\tilde{O}(|I|^{2\omega/(\omega+1)})$, where $\omega$ is the matrix multiplication exponent (the submodular width of the triangle is $3/2$).

\subparagraph*{Conjectures in Fine-Grained Complexity} Our lower bounds will be based on the following popular conjectures in fine-grained complexity. To state the conjectures, it will be helpful to define the {\em $k$-clique problem over a semiring $\sigma$}: given an undirected graph $G = (V,E)$ where each edge has a weight in the domain of the semiring, we are asked to compute the semiring-product over all the $k$-cliques in $G$, where the weight of each clique is the semiring-sum of clique edge weights.

\begin{definition}[Boolean $k$-Clique Conjecture]
 There is no real $\epsilon >0$ such that computing the $k$-clique problem (with $k \geq 3$) over the Boolean semiring in an (undirected) $n$-node graph requires time $O(n^{k-\epsilon})$ using a combinatorial algorithm.
 \end{definition}

\begin{definition}[Min-Weight $k$-Clique Conjecture]
 There is no real $\epsilon >0$ such that computing the $k$-clique problem (with $k \geq 3$) over the tropical semiring in an (undirected) $n$-node graph with integer edge weights can be done in time $O(n^{k-\epsilon})$.
\end{definition}

When $k=3$, min-weight 3-clique is equivalent to the All-Pairs Shortest Path (APSP) problem under subcubic reductions. The Min-Weight Clique Conjecture assumes the Min-Weight $k$-Clique conjecture for every integer $k \geq 3$ (similarly for the Boolean Clique Conjecture).

\section{The Clique Embedding Power} \label{sec:emb}

In this section, we define the clique embedding power,  a quantity central to this paper. 

\subsection{Graph Embeddings}

We introduce first the definition of embedding a graph $G$ to a hypergraph $\mH$, first defined by Marx~\cite{Marx13,Marx10}. We say that two sets of vertices $X,Y \subseteq V(\mH)$ {\em touch} in $\mH$ if either $X \cap Y \neq \emptyset$ or there is a hyperedge $e \in E(\mH)$ that intersects both $X$ and $Y$. We say a hypergraph is connected if its underlying clique graph is connected.

\begin{definition}[Graph Embedding]
Let $G$ be an undirected graph, and $\mH$ be a hypergraph. An {\em embedding} from $G$ to $\mH$, denoted $G \mapsto \mH$, is a mapping $\psi$ that maps every vertex $v \in V(G)$ to a non-empty subset $\psi(v) \subseteq V(\mH)$ such that the following hold:
\begin{enumerate}
\item $\psi(v)$ induces a connected subhypergraph; 
\item if $u,v \in V(G)$ are adjacent in $G$, then $\psi(u), \psi(v)$ {\em touch} in $\mH$.
\end{enumerate}
\end{definition}

It will often be convenient to describe an embedding $\psi$ by the reverse mapping $\psi^{-1}(x) = \{i \mid x \in \psi(i) \}$, where $x$ is a vertex  in $V(\mH)$. Given an embedding $\psi$ and a vertex $v \in V(\mH)$, we define its {\em vertex depth} as $d_\psi(v) = |\psi^{-1}(v)|$. For a hyperedge $e \in E(\mH)$, we define its {\em weak edge depth} as $d_{\psi}(e) = |\{v \in V(G) \mid  \psi(v) \cap e \neq \emptyset \}|$, i.e., the number of vertices of $G$ that map to some variable in $e$. Moreover, we define the {\em edge depth} of $e$ as $d^+_{\psi}(e) = \sum_{v \in e} d_\psi(v)$.

The {\em weak edge depth} of an embedding $\psi$ can then be defined as $\mathsf{wed}(\psi) = \max_e d_\psi(e)$, and the edge depth as $\mathsf{ed}(\psi) = \max_e d^+_\psi(e)$.  Additionally, we define as $\mathsf{wed}(G \mapsto \mH)$ the minimum weak edge depth of any embedding $\psi$ from $G$ to $\mH$. Similarly for  $\mathsf{ed}(G \mapsto \mH)$. It is easy to see that $\mathsf{wed}(G \mapsto \mH) \leq \mathsf{ed}(G \mapsto \mH)$.

It will be particularly important for our purposes to think about embedding the $k$-clique graph $C_k$ to an arbitrary hypergraph $\mH$. In this case, it will be simpler to think of the vertices of $G$ as the numbers $1, \dots, k$ and the embedding $\psi$ as a mapping from the set $\{1, \dots, k\}$ to a subset of $V(\mH)$. We can now define the following quantity, which captures how well we can embed a $k$-clique to $H$ for an integer $k \geq 3$: 
$$\mathsf{emb}_k(\mH) :=  \frac{k}{\mathsf{wed}(C_k \mapsto \mH)}.$$

\begin{example}
Consider the hypergraph $\mH$ with the following hyperedges:
$$ \{x_1, x_2, x_3\}, \{x_1, y\}, \{x_2, y\}, \{x_3, y\}$$
We can embed the $5$-clique into $\mH$ as follows:
$$ 1 \rightarrow \{ x_1\}, 2 \rightarrow \{x_2\}, 3  \rightarrow \{x_3\}, 4,5 \rightarrow \{y\}. $$
It is easy to check that this is a valid embedding, since, for example, $1,4$ touch at the edge $\{x_1, y\}$. Moreover, $\mathsf{wed}(C_5 \mapsto G) = 3$, hence $\mathsf{emb}_5(G) = 5/3$.
\end{example}

\subsection{Embedding Properties}

In this part, we will explore how $\mathsf{wed}(C_k \mapsto \mH)$ and $\emb_k(\mH)$ behave as a function of the size of the clique $k$. We start with some basic observations.

\begin{proposition}\label{prop:emb:properties}
For any hypergraph $\mH$ and integer $k \geq 3$:
\begin{enumerate}
\item $\mathsf{wed}(C_k \mapsto \mH) \leq k$.
\item $\mathsf{wed}(C_k \mapsto \mH)  \leq \mathsf{wed}(C_{k+1} \mapsto H) \leq \mathsf{wed}(C_k \mapsto \mH)  +1.$
\item If $k = m \cdot n$, where $k,m,n \in \mathbb{Z}_{\geq 0}$, then $\mathsf{emb}_k(\mH) \geq \mathsf{emb}_m(\mH)$.
\end{enumerate}
\end{proposition}

\begin{proof}
(1) We define an embedding $\psi$ from a $k$-clique where $\psi(i) = V(\mH)$ for every $i =1, \dots, k$. It is easy to see that $\psi$ is an embedding with weak edge depth $k$.

\smallskip

(2) For the first inequality, take any $\psi_{k+1}$, we can construct a $\psi_{k}$ by only preserving the mapping $\psi_{k+1}$ for $[k]$. Then, for any $e \in E(\mH)$, we have 
$$
\{y \in V(C_k) \mid \psi_k(y) \cap e \neq \emptyset\}  \subseteq \{y \in V(C_{k+1}) \mid \psi_{k+1}(y) \cap e \neq \emptyset\}
$$
Thus, 
$$
\mathsf{wed}(C_k \mapsto \mH) \leq \mathsf{wed}(\psi_{k}) := \max_{e \in E(\mH)} d_{\psi_k}(e) \leq \mathsf{wed}(\psi_{k+1}).
$$
For the second inequality, 
take any $\psi_{k}$, we construct a $\psi_{k+1}$ by preserving the mapping $\psi_{k}$ and $\psi_{k+1}$ maps $k+1$ to $V(\mH)$. Then, for any $e \in E(\mH)$, we have 
$$
d_{\psi_{k+1}}(e) = d_{\psi_k}(e) + 1
$$
so $\mathsf{wed}(\psi_{k+1}) = \mathsf{wed}(\psi_{k}) + 1$ and in particular, we can take $\psi_{k}$ such that 
$$
\mathsf{wed}(\psi_{k+1}) \leq  \mathsf{wed}(C_k \mapsto \mH) + 1
$$
which implies that 
$$
\mathsf{wed}(C_{k+1} \mapsto \mH) \leq \mathsf{wed}(C_k \mapsto \mH)+1
$$
\smallskip
(3) Suppose $\psi$ is the embedding that achieves $\mathsf{emb}_m(\mH)$ for $C_m$. It suffices to construct an embedding $\psi'$ for $C_k$ which achieves the same quantity $\mathsf{emb}_m(\mH)$. To do so, we simply bundle every $n$ vertices in $C_k$ to be a ``hypernode''. That is, label the bundles as $b_1, \dots, b_n$. and $\psi'(v) = \psi(i)$ if and only if $v \in B_i$. The embedding power given by $\psi'$ is then 
$$\frac{k}{\mathsf{wed}(\psi')} = \frac{mn}{\mathsf{wed}(\psi)n} = \frac{m}{\mathsf{wed}(\psi)} = \mathsf{emb}_m(\mH).$$
\end{proof}

% \begin{proposition}
% For any $H$, $\mathsf{wd}(C_k \rightarrow H) \geq k/\rho^*(H)$.
% \end{proposition}
%
%\begin{proof}
%Let  $w = \mathsf{wd}(C_k \rightarrow H)$. This means that there exists an embedding $\psi$ from $C_k$ to $H$ with width $w$. We will use $\psi$ to construct a fractional vertex packing with total weight $\geq k/w$; since any fractional edge cover must exceed this by LP duality, this proves the desired statement. 
%
%For every $i=1, \dots, k$, pick some vertex $x \in \psi(i)$; we denote this by $\psi^+(i)$. Now, simply assign to variable $x$ the weight 
%%
%$$ v_x = \frac{|i \mid \psi^+(i)=x|}{w}$$ 
%First, note that the total weight of this assignment is $\sum_{x} v_x = k/w$, since each $i =1, \dots, k$ is accounted exactly once in the summation. It remains to show that $\{v_x\}_x$ is a valid edge packing for $H$. Indeed, consider any hyperedge $e \in E(H)$. Then, $\sum_x {|i \mid \psi^+(i)=x|}$ can never exceed $w$, otherwise we would assigned more than $w$ different integers to $e$. Thus, $\sum_{x \in e} v_x \leq 1$.
%\end{proof}

The first item of the above proposition tells us that $\emb_k(\mH) \geq 1$ for any $k$.
But how does $\emb_k(\mH)$ behave as $k$ grows? We next show that $\emb_k(\mH)$ is always upper bounded by the submodular width of $\mH$.

\begin{lemma} \label{lem:decomp}
Let $\mH$ be a hypergraph. Take an embedding $\psi: C_k \mapsto \mH$. Let $(\htree, \chi)$ be a tree decomposition of $\mH$. Then, there exists a node $t \in T$ such that for every $i=1, \dots, k$, $\psi(i) \cap \chi(t) \neq \emptyset$.
\end{lemma}

%In other words, there exists a bag $B_t$ that "touches" all variables of the $k$-clique. 

\begin{proof}
For $i=1, \dots, k$, let $\htree_i$ be the subgraph of $\htree$ that includes all nodes $t \in V(\htree)$ such that $\psi(i) \cap \chi(t)  \neq \emptyset$. 

We first claim that {\em $\htree_i$ forms a tree}. To show this, it suffices to show that $\htree_i$ is connected. Indeed, take any two nodes $t_1, t_2$ in $\htree_i$. This means that there exists $x_1 \in \chi(t_1) \cap \psi(i)$ and $x_2 \in \chi({t_2}) \cap \psi(i)$. Since $x_1, x_2 \in \psi(i)$ and $\psi(i)$ induces a connected subgraph in $\mH$, there exists a sequence of vertices $x_1 =  z_1, \dots, z_k = x_2$, all in $\psi(i)$, such that every two consecutive vertices belong to an edge of $\mH$. Let $S_1, \dots, S_k$ be the trees in $\htree$ that contain $z_1, \dots, z_k$ respectively. Take any two consecutive $z_i, z_{i+1}$: since they belong to the same edge, there exists a bag that contains both of them, hence $S_i, S_{i+1}$ intersect. This means that there exists a path between $t_1, t_2$ in $\htree$ such that every node is in $T_i$.

Second, we claim that {\em any two $\htree_i, \htree_j$ have at least one common vertex}. Indeed, $\psi(i)$, $\psi(j)$ must touch in $\mH$. If there exists a variable $x \in \psi(i) \cap \psi(j)$, then any vertex that contains $x$ is a common vertex between $\htree_i, \htree_j$. Otherwise, there exists $x \in \psi(i)$, $y \in \psi(j)$ such that $x,y$ occur together in a hyperedge $e \in E(\mH)$. But this means that some node $t \in \htree$ contains both $x,y$, hence $\htree_i, \htree_j$ intersect at $t$.

Finally, we apply the fact that a family of subtrees of a tree satisfies the {\it Helly property}~\cite{Heggernes05}, i.e. {\em a collection of subtrees of a tree has at least one common node if and only if
every pair of subtrees has at least one common node.} Indeed, the trees $\htree_1, \dots, \htree_k$ satisfy the latter property, so there is a vertex $t$ common to all of them. Such $t$ has the desired property of the lemma. 
\end{proof}

We can now state the following \autoref{prop:lowerbound} on the embedding power of a hypergraph. %A similar result for $\adw$ is stated and proven in \autoref{thm:adw} in Appendix~\ref{app:emb} of the full version of the paper \cite{fan2023finegrained}.

\begin{theorem}\label{prop:lowerbound}
For any hypergraph $\mH$ and integer $k \geq 3$, the following holds:
$$\mathsf{wed}(C_k \mapsto \mH) \geq \frac{k}{\mathsf{subw}(\mH)}$$
\end{theorem}

\begin{proof}
Let $\mathsf{wed}(C_k \mapsto \mH)=\alpha$. Then, there is an embedding $\psi: C_k \mapsto \mH$ with weak edge depth $\alpha$. We will show that $\mathsf{subw}(\mH) \geq k/\alpha$.

First, we define the following set function over subsets of $V(\mH)$: for any $S \subseteq V(\mH)$, let $\mu(S) = |\{i \mid \psi(i) \cap S \neq \emptyset\}|/\alpha$. This is a coverage function, and hence it is a submodular function. It is also edge-dominated, since for any hyperedge $e$, we have $\mu(e) =   |\{i \mid \psi(i) \cap e \neq \emptyset\}|/\alpha \leq 1$.

Now, consider any decomposition $(\htree, B_t)$ of $\mH$. From Lemma~\ref{lem:decomp}, there is a node $t \in \htree$ such that or every $i=1, \dots, k$, $\psi(i) \cap B_t \neq \emptyset$. Hence, $\mu(B_t) =  |\{i \mid \psi(i) \cap B_t \neq \emptyset\}|/\alpha = k/\alpha$. Thus, the submodular width of the decomposition is at least $k/\alpha$.
\end{proof}

The above result tells us that $\mathsf{emb}_k(\mH) \leq \mathsf{subw}(\mH)$ for any $k \geq 3$. Hence, taking the supremum of  $\mathsf{emb}_k(\mH)$ for $k \geq 3$ is well-defined since the set is bounded. This leads us to the following definition: 

\begin{definition}[Clique Embedding Power]
Given a hypergraph $\mH$, define  the {\em clique embedding power} of $\mH$ as  
$$\mathsf{emb}(\mH) := \sup_{k \geq 3} \mathsf{emb}_k(\mH) = \sup_{k \geq 3} \frac{k}{\mathsf{wed}(C_k \mapsto \mH)}.$$
\end{definition}

The following is immediate:

\begin{corollary}
For any hypergraph $\mH$, $1 \leq \emb(\mH) \leq \subw(\mH)$.
\end{corollary}

For the connection between edge depth width and adaptive width, we have the following theorem analogous to Theorem~\ref{prop:lowerbound}. The proof can be found in~\cite{arxiv}.

\begin{theorem} \label{thm:adw}
For any hypergraph $\mH$, the following holds:
$$\mathsf{ed}(C_k \mapsto \mH) \geq \frac{k}{\mathsf{adw}(\mH)}$$
\end{theorem}

%\begin{proof}
%Let $\mathsf{ed}(C_k \mapsto \mH)=\alpha$. Then, there is an embedding $\psi: C_k \mapsto \mH$ with edge depth $\alpha$. We will show that $\mathsf{adw}(\mH) \geq k/\alpha$.
%
%First, we define the following function over subsets of $V(\mH)$: for any $S \subseteq V(\mH)$, let $\mu(S) = \sum_{v \in S} d_\psi(v)/\alpha$. This function forms a fractional independent set, since  for any hyperedge $e$, we have $\mu(e) = \sum_{v \in e } d_\psi(v)/\alpha = d^+_\psi(e)/\alpha \leq 1$.
%
%Now, consider any decomposition $(\htree, \chi)$ of $\mH$. From Lemma~\ref{lem:decomp}, there is a node $t \in T$ such that or every $i=1, \dots, k$, $\psi(i) \cap \chi(t) \neq \emptyset$. Hence, $\mu(B_t) =  \sum_{v \in B_t} d_\psi(v) /\alpha \geq k/\alpha$. Thus, the adaptive width of the decomposition is at least $k/\alpha$.
%\end{proof}

%Thus, $\mathsf{emb}(H) \leq \mathsf{subw}(H)$,  which means that the conditional lower bound can never exceed the submodular width bound, a good sanity check! Notice that for the case where $\mathsf{emb}(H) = \mathsf{subw}(H)$, we have a tight upper bound! 

\section{Lower Bounds} 
\label{sec:lower:bounds}

In this section, we show how to use the clique embedding power to obtain conditional lower bounds for SumProduct problems. Our main reduction follows the reduction used in~\cite{Marx13}, but also has to account for constructing the appropriate values for the semiring computations.

\begin{theorem}\label{thm:reduction}
For any hypergraph $\mH$ and semiring $\sigma$, if  $\mathsf{SumProd}\langle \sigma, \mH \rangle$ can be solved in time $O(|I|^c)$ with input $I$, then $k$-Clique over $\sigma$ can be solved in time $O(n^{c \cdot \mathsf{wed}(C_k \mapsto \mH)})$ where $n$ is the number of vertices.
\end{theorem}

\begin{proof}
We will show a reduction from the $k$-clique problem with $n$ vertices over a semiring $\sigma$ to  $\mathsf{SumProd}\langle \sigma, \mH \rangle$. Without loss of generality, we will assume that the input graph $G$ to the $k$-clique problem is $k$-partite, with partitions $V_1, \dots, V_k$. Indeed, given any graph $G = (V,E)$ where $V = \{v_1, v_2, \dots, v_n\}$, consider the $k$-partite graph $G^k = (V^k,E^k)$ where $V^k = \{v_i^j \mid 1 \leq i \leq n, 1 \leq j \leq k\}$ and for any two vertices $v_i^j, v_p^q \in V^k$, $\{v_i^j, v_p^q\} \in E^k $ iff $\{v_i, v_p\} \in E$ and $j \neq q$. Then there is a one-to-one mapping from a $k$-clique in $G$ to a $k$-clique in $G^k$.

Let $\psi$ be an embedding from $C_k$ to $\mH$ that achieves a weak edge depth $\lambda = k / \mathsf{emb}_k(\mH) $. As we mentioned before, it is convenient to take $V(C_k) = \{1, \dots, k\}$. We now construct the input instance $I$ for $\mathsf{SumProd}\langle \sigma, \mH \rangle$.  More explicitly, the task is to construct the function $R_e$ for each hyperedge $e \in E$. 

To this end, we first assign to each pair $\{ i,j \}: i \neq j, i,j, \in \{1,2,\dots, k\}$ a hyperedge $\theta(\{ i,j \}) = e \in E(\mH)$ satisfying the following conditions: $\psi(i) \cap e \neq \emptyset$ and $\psi(j) \cap e \neq \emptyset$. Such an $e$ must exist by the definition of an embedding. If there is more than one hyperedge satisfying the condition, we arbitrarily choose one.

For every variable $x \in V(\mH)$, let $\psi^{-1}(x)$ be the subset of $\{1, \dots, k\}$ mapping to $x$.  Then, we define the domain $\mathsf{Dom}(x_i)$ of each variable $x_i$ in the input instance as vectors over ${[n]}^{|\psi^{-1}(x_i)|}$. Let $S_e = \{i \in [k] \mid  \psi(i) \cap e \neq \emptyset \}$. Note that $|S_e| = d_{\psi}(e) \leq \lambda$. Also, note that $\psi^{-1}(x) \subseteq S_e$ for all $ x \in e$. Then, we compute all cliques in the graph $G$ between the partitions $V_i, i \in S_e$; these cliques will be of size  $|S_e|$  and can be computed in running time $O(n^\lambda)$ by brute force.  

For every clique $\{a_i \in V_i \mid i \in S_e \}$, let $t$ be the tuple over $\prod_{i \in S_e} \mathsf{Dom}(x_i)$ such that its value at position $x$ is  $\langle a_i \mid i \in \psi^{-1}(x) \rangle$. Then, we set
$$ R_e(t) = \mathbf{1} \otimes \bigotimes_{\{ i,j \}: \theta(\{ i,j \}) = e} w(\{i,j\}).$$
In other words, we set the value as the semiring product of all the weights between the edges $\{a_i, a_j\}$ in the clique whenever the pair $\{ i,j \}$ is assigned to the hyperedge $e$. All the other tuples are mapped to $\mathbf{0}$. By construction, the size of the input is $|I| = O(n^\lambda)$.

%The relation associated to a hyperedge $e  \in V(H)$ is constructed as in the proof of Theorem~\ref{thm: boolean} except that we now take care of the weights of the edges.  The idea is that since we associate each pair  $i,j \in \{1,2,\dots, k\}$ a hyperedge $e \in E(H)$, the weight of each edge will be added only once in the $k$-clique. 

%The function $R_e$ associated with a hyperedge $e \in V(H)$ is constructed as follows. 

\smallskip

We now show that the two problems will return exactly the same output. To show this, we first show that there is a bijection between $k$-cliques in $G$ and the solutions of the SumProduct instance.

$\Leftarrow$ Take a clique $\{ a_1, \dots, a_k\}$ in $G$. We map the clique to the valuation $v(x) = \langle a_i \mid i \in \psi^{-1}(x) \rangle$. This valuation is a solution to the SumProduct problem, since any subset of  $\{ a_1, \dots, a_k\}$ forms a sub-clique. Hence for any hyperedge $e$, $R_e(v(\bx_e)) \neq \mathbf{0}$. 

\medskip

$\Rightarrow$ Take a valuation $v$. Consider any $i \in \{1, \dots, k\}$ and consider any two variables $x,y \in \psi(i)$ (recall that $\psi(i)$ must be nonempty). Recall that $x,y \in V(\mH)$. We claim that the $i$-th index in the valuation $v(x), v(y)$ must take the same value, which we will denote as $a_i$; this follows from the connectivity condition of the embedding. Indeed, since $x,y\in \psi(i)$, there exists a sequence of hyperedges $e_1, e_2, \dots, e_m$ where $m \geq 1$ such that $e_j \cap e_{j+1} \neq \emptyset$ for $1 \leq j \leq m-1$ and $x\in e_1$, $y \in e_m$.  By the construction, the $i$-th index in $v(x)$ will then ``propagate'' to that in $v(y)$. This proves the claim. It then suffices to show that $\{ a_1, \dots, a_k\}$ is a clique in $G$. Indeed take any $i,j \in \{1, \dots, k\}$. Since $i$ and $j$ are adjacent as two vertices in $C_k$, we know $\psi(i)$ and $\psi(j)$ touch. Therefore, there exists a hyperedge $e$ that contains some $x \in \psi(i)$ and $y \in \psi(j)$. But this means that $\{a_i, a_j\}$ must form an edge in $G$.

\medskip

We next show that the semiring product of the weights in the clique has the same value as the semiring product of the corresponding solution. Indeed:
$$ \bigotimes_{e \in E} R_e(v(\bx_e)) = \mathbf{1} \otimes \bigotimes_{e \in E} \bigotimes_{\{ i,j \}: \theta(\{ i,j \}) = e} w(\{i,j\}) =  \bigotimes_{\{ i,j \}: i \neq j} w(\{i,j\})$$
where the last equality holds because each edge of the $k$-clique is assigned to exactly one hyperedge of $\mH$.

The above claim together with the bijection show that the output will be the same; indeed, each the semiring sums will go over exactly the same elements with the same values.
\medskip

To conclude the proof, suppose that $\mathsf{SumProd}\langle \sigma, \mH \rangle$ could be answered in time $O(|I|^c)$ for some $c \geq 1$. This means that we can solve the $k$-clique problem over $\sigma$ in time $O(n^\lambda + n^{c \lambda}) = O(n^{c \cdot \mathsf{wed}(C_k \mapsto \mH)})$.
\end{proof}

As an immediate consequence of Theorem~\ref{thm:reduction}, we can show the following lower bound.

\begin{proposition}
Under the Min-Weight $k$-Clique conjecture, $\mathsf{SumProd}\langle \mathsf{trop}, \mH \rangle$ (and thus $\mathsf{SumProd}\langle \mH \rangle$) cannot be computed in time $O(|I|^{\mathsf{emb}_k(\mH)-\epsilon})$ for any constant $\epsilon >0$.
\end{proposition}

\begin{proof}
Indeed, if  $\mathsf{SumProd}\langle \mathsf{trop}, \mH \rangle$ can be computed in time $O(|I|^{\mathsf{emb}_k(\mH)-\epsilon})$ for some constant $\epsilon >0$, then by Theorem~\ref{thm:reduction} the $k$-Clique problem over the tropical semiring can be solved in time $O(n^{(\mathsf{emb}_k(\mH)-\epsilon) \cdot \mathsf{wed}(C_k \mapsto \mH)}) = O(n^{k-\delta})$ for some $\delta>0$. However, this violates  the Min-Weight $k$-Clique conjecture.
\end{proof}

Similarly, we can show the following:

\begin{proposition}
Under the Boolean $k$-Clique conjecture, $\mathsf{SumProd}\langle \sigma_{\mathbb{B}}, \mH \rangle$ (and thus $\mathsf{SumProd}\langle \mH \rangle$) cannot be computed via a combinatorial algorithm in time $O(|I|^{\mathsf{emb}_k(\mH)-\epsilon})$ for any constant $\epsilon >0$.
\end{proposition}

%This generalization nicely demonstrates the advantage of viewing the evaluation of conjunctive query as a sum-of-product computation~\cite{FAQ} and should be of independent interest.

%\begin{theorem}
%For any hypergraph $H$, if the conjunctive query $Q_H$ over the tropical semiring can be solved in time $O(|I|^c)$, then minimum weight $k$-Clique can be solved in time $O(n^{c \cdot \mathsf{wed}(C_k \mapsto H)})$.
%\end{theorem}
%
%\begin{proof}
%Without loss of generality, we assume the input $G$ to the minimum weight $k$-clique problem is $k$-partite and $V(G) = (V_1, V_2, \dots, V_k)$. The reasoning is the same as in the proof of Theorem~\ref{thm: boolean}, except that now the edge $\{v_i^j, v_p^q\}$ will have the same weight in $E^k(G)$ to that of the edge $(v_i, v_p)$ in $E(G)$.
%%It remains to show how to do the weight assignment on the tuples of $I$. Note that the embedding guarantees that every pair $i,j$ is "contained" in at least one hyperedge, meaning there exists some hyperedge $e$ that contains $a_i, a_j$ if they form an edge in $E$ as values. We now assign each pair $i,j$ to exactly one hyperedge (if there are more than one, we arbitrarily pick one), and then each tuple sums the weights of the pairs assigned to it. If a hyperedge has no assigned pairs, its weight is simply 0.
%%It is easy to see that the weight of an output tuple for $H_Q$ is exactly the sum of the weights of the clique, since each edge in the clique is accounted exactly once.
%\end{proof}

The above two results imply that to obtain the best lower bound, we need to find the clique size with the largest $\mathsf{emb}_k(\mH)$. However, the function $k \mapsto \mathsf{emb}_k(\mH)$ is really intriguing. It is not clear whether in the definition supremum is ever needed, i.e., whether there exists a hypergraph where the embedding power is achieved in the limit. %It is not even clear from the definition whether computing the embedding power for a hypergraph is decidable.

In \autoref{sec: decidability}, we show that for every hypergraph $\mH$, there exists a natural number $k$ such that $\mathsf{emb}(\mH) = \mathsf{emb}_k(\mH)$. We also demonstrate how to compute $\mathsf{emb}(\mH)$ through a MILP and locate the complexity of computing the embedding power within the class 2-$\mathsf{EXPTIME}$ (double exponential time). The insight of our method is that, instead of computing the ``integral'' embedding power, one can consider the ``fractional'' embedding power and then recover the ``integral'' one by letting the clique size $k$ to be sufficiently large.

\section{Decidability of the Clique Embedding Power}\label{sec: decidability}

To illustrate the algorithm for computing $\emb(\mH)$, it is instructive to first show how to compute $\emb_k(\mH)$ for a fixed clique size $k$.

\subsection{An Integer Linear Program for $\mathsf{wed}(C_k \mapsto \mH)$}

The following ILP formulation computes the minimum weak edge depth $w = \mathsf{wed}(C_k \mapsto \mH)$. 
%
% We use the shorthand $e \cap i \neq \emptyset$ for an hyperedge $e \in \mathcal{E}$ to mean that there exists at least one common vertex in $e$ and the vertices set represented by $i \in [N]$. 
\begin{equation}\label{MIP k}
\begin{array}{llllllll}
\displaystyle \min & \multicolumn{3}{l}w \\
\textrm{s.t.}
& & &\displaystyle \sum\limits_{S \subseteq V} x_S &= &k \\
& & &x_S &= &0 &\forall S \subseteq V &\textrm{ where } \textrm{$S$ is not connected}\\
& & & \min \{ x_{S}, x_T\} &=& 0 &  \forall S,T \subseteq V &\textrm{ where $S,T$ do not touch}  \\
%& x_i &+ &k\cdot y_i &\leq &k  \\
%& x_j &+ &k\cdot y_j &\leq &k & \forall i,j \in [N]  &\textrm{ where }  \textsf{touched}(i,j)  = \texttt{False}\\
%& y_i &+ &y_j &\geq &1  \\
\displaystyle & & & \sum\limits_{S \subseteq V: e \cap S \neq \emptyset } x_S  & \leq & w &\forall e \in \mathcal{E} \\
& & &x_S &\in &\mathbb{Z}_{\geq 0} &\forall S \subseteq V 
% & & & y_i &\in &\{0,1\}  &\forall i \in [N]\\
\end{array}
\end{equation}

Each integer variable $x_S$, $S \subseteq V$, indicates how many vertices in $C_k$ are assigned to the subset $S$. For example, if $x_{\{1,2\}} = 3$, this means in the embedding $\psi$, three vertices are mapped to the subset $\{1,2\} \subseteq V$. It is sufficient to record only the number of vertices in $C_k$ because of the symmetry of the clique. That is, since any two vertices are connected in $C_k$, one can arbitrarily permute the vertices in $C_k$ so that the resulting map $\psi'$ is still an embedding (given $\psi$ is).
Moreover, since the clique size $k$ is fixed, to compute $\mathsf{emb}_k(\mH)$ it suffices to minimize $w$. 

%The function  \textsf{connected()} takes an integer $i \in [N-1]$ and outputs \texttt{True} if the corresponding subset of $V$ is not connected in $H$, and \texttt{False} otherwise. Note that could be exponentially many $i$, therefore subsets of vertices in $H$, such that $\textsf{connected($i$)} = \texttt{False}$, leading to exponentially many (w.r.t. to the size of $|H|$) linear conditions of this kind. 

Observe that the condition $ \min \{ x_{S}, x_T\} = 0 $ is not a linear condition. To encode it as such, we perform a standard transformation. We introduce a binary variable $y_S$ for every set $S \subseteq V$. Then, we can write it as
%
%The binary variables $y_i$'s, $1 \leq i \leq N-1$, are created so as to write the ``touch'' condition in the form of linear relation. Recall that $[N] = \{1,2,\dots,N\}$. 
%The function \textsf{touched()} takes two integers $i,j \in [N-1]$ and outputs \texttt{True} if the two corresponding subsets of $V$ does not touch in $H$, and \texttt{False} otherwise. Note that the conditions 
\begin{equation*}
\left\{
\begin{array}{lllll}
	& x_S &+ &k\cdot y_S &\leq k  \\
	& x_T &+ &k\cdot y_T &\leq k \\
	& y_S &+ &y_T &\geq 1 \\
\end{array} 
\right.
\end{equation*}
%
%ensure any two vertices in $C_k$ touch in $H$. This is because when the subsets corresponding to $i$ and $j$ does not touch, 
Indeed, since $y_S$ and $y_T$ are binary variables, at least one of them is $1$. Without loss of generality assume $y_S$ = 1. Then $x_S = 0$ since $x_S \in \mathbb{Z}_{\geq 0}$. Therefore two subsets that do not touch cannot both be chosen in the embedding.

%Finally, note that the conditions $e \cap i \neq \emptyset$ can be checked efficiently. A moment of reflection should convince the reader that this integer programming correctly computes the quantity $\mathsf{wed}(C_k \mapsto H)$.

%To solve this integer programming, one could loop over all possible $\{0,1\}$ assignments of $y_i$'s. This reduces the problem to a linear programming problem which can be solved efficiently. 

\subsection{A Mixed Integer Linear Program for $\emb(\mH)$}

The above ILP construction does not directly yield a way to compute the clique embedding power, since the latter is defined to be the supremum for all $k$. %Also, when $k$ is large, solving the mixed integer programming (\ref{MIP k}) is computationally infeasible.
\begin{figure}[h!]
\centering
\captionsetup{justification=centering}
\includegraphics[scale=0.4]{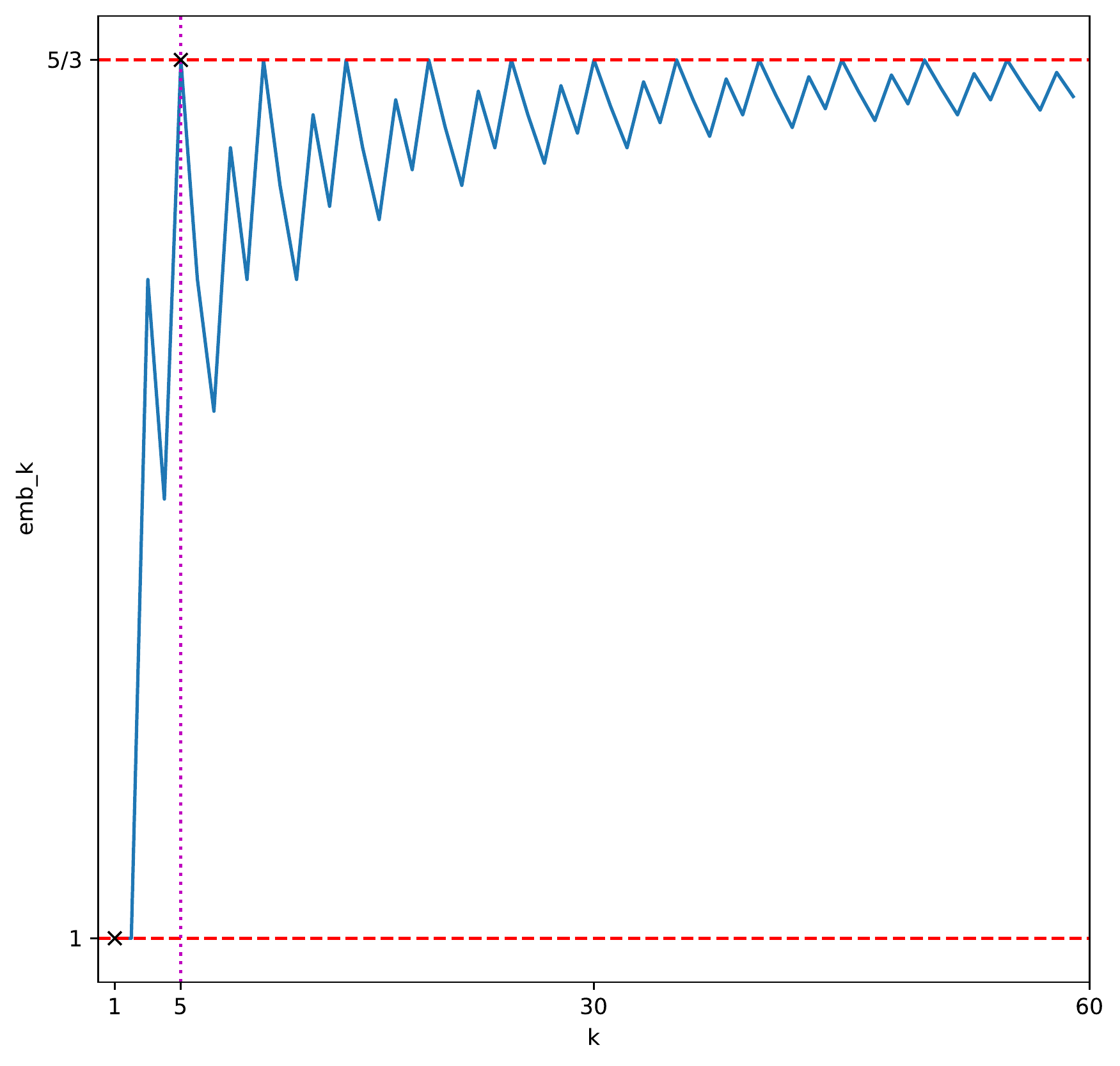}
\caption{$\mathsf{emb}_k(\mH)$ for the 6-cycle}
\label{fig:emb}
\end{figure}

As alluded before, the behavior of $\mathsf{emb}_k(\mH)$ as a function of $k$ is non-trivial (and certainly not monotone). Figure~\ref{fig:emb} depicts how the clique embedding power changes with respect to different clique sizes for the 6-cycle, where the horizontal line represents the clique size.

To compute the supremum, the key idea is to change the integer variables $x_S$ to be continuous (so they behave as fractions) and "normalize" the clique size $k$ to 1. Specifically, we can write the following mixed integer linear program (MILP).
\begin{equation}\label{MIP}
\begin{array}{llllllll}
\displaystyle \min & \multicolumn{3}{l}w \\
\textrm{s.t.}
& & &\displaystyle \sum\limits_{S \subseteq V} x_S &= &1 \\
& & &x_S &= &0 &\forall S \subseteq V &\textrm{ where } \textrm{$S$ is not connected}\\
& & & \min \{ x_{S}, x_T\} &=& 0 &  \forall S,T \subseteq V &\textrm{ where $S,T$ do not touch}  \\
%& x_i &+ &k\cdot y_i &\leq &k  \\
%& x_j &+ &k\cdot y_j &\leq &k & \forall i,j \in [N]  &\textrm{ where }  \textsf{touched}(i,j)  = \texttt{False}\\
%& y_i &+ &y_j &\geq &1  \\
\displaystyle & & & \sum\limits_{S \subseteq V: e \cap S \neq \emptyset } x_S  & \leq & w &\forall e \in \mathcal{E} \\
& & &x_S &\in &\mathbb{R}_{\geq 0} &\forall S \subseteq V 
\end{array}
\end{equation}

%There are two changes from MIP (\ref{MIP k}) to MIP (\ref{MIP}). Firstly, the variables $x_i$'s now take non-negative real values. Secondly, every occurrence of $k$ changes to 1. 

\begin{proposition} \label{prop:milp}
Let $w^*$ be the optimal solution of MILP~\eqref{MIP}. Then, $\emb(\mH) = 1/w^*$. Additionally, there exists an integer $K \geq 3$ such that $\emb(\mH) = \emb_K(\mH)$.
\end{proposition}

\begin{proof}
%Let $\mathsf{emb}^* := \frac{1}{w}$. We claim that $\mathsf{emb}(H) = \mathsf{emb}^*$. 
We first show for any $k$, $\mathsf{emb}_k(\mH) \geq 1/w^*$. Indeed, any embedding $\psi : C_k \rightarrow \mH$ determines the values of the variables $x_S$ in MILP~\eqref{MIP k}. Let $\hat{x_S} = \frac{x_S}{k}$ and $\hat{w} = \frac{w}{k}$ be an assignment of the variables in MILP~\eqref{MIP}. It is easy to see that this is a feasible assignment. Thus, $w^* \leq \mathsf{wed}(C_k \mapsto \mH)/k$. Therefore $\mathsf{emb}_k(\mH) = k/\mathsf{wed}(C_k \mapsto \mH) \leq 1/w^*$.

Next, observe that $\mathsf{emb}(\mH)$ is a rational number. In fact, the solution $w^*$ for MILP~\eqref{MIP} is a rational number, since every constant is a rational number~\cite{Schrijver99}. Let $K$ be the least common multiplier of their denominators of the fractions in the set $\{x_S\}$. Then, the assignment $K \cdot x_S, K \cdot w$  is a feasible solution for MILP~\eqref{MIP k} for $k =K$. This implies that $K \cdot w^* \geq \mathsf{wed}(C_K \mapsto \mH)$, so $\emb_K(\mH) \geq 1/w^*$. 

Thus, we have shown that $1/w^*$ is an upper bound for $\{\emb_k(\mH)\}_k$, but also $\emb_K(\mH) = 1/w^*$. Hence, $\emb(\mH) = \emb_K(\mH) = 1/w^*$.
\end{proof}

This leads to the following theorem (whose proof can be found in \cite{arxiv}).

\begin{theorem}\label{thm: fractional emb}
The problem of computing ${\emb}(\mH)$ for a hypergraph $\mH$ is in 2-$\mathsf{EXPTIME}$ and, in particular, is decidable.
\end{theorem}

% \begin{proof}
% Proposition~\ref{prop:milp} shows that to compute ${\emb}(\mH)$ it suffices to solve MILP~\eqref{MIP}. To solve the MILP, we can sequentially fix the assignments of the binary variables $y_S$ and then solve a linear program. Note that the remaining linear program might have exponentially many conditions in the size of $\mH$, since the number of variables is $2^{|V|}$. The proof is then completed by observing the number of assignments of $y_S$ are doubly exponential in $\mH$.
% \end{proof}

Unfortunately, our method does not yield an upper bound on how large the $K$ in Proposition~\ref{prop:milp} might be. There is no reason that $K$ cannot be very large, e.g.\ doubly exponential to the size of $\mH$. Some knowledge of that could be very useful in computing the clique embedding power. For example, one can compute all the embeddings from $C_k$, for $k$ not greater than the upper bound, and output the one with the largest embedding power. 
The best-known upper bound we have so far is the following. The proof can be found in~\cite{arxiv}.

\begin{proposition} \label{prop:Kbound}
For any hypergraph $\mH$, there is a constant $K= O((2^{|V|})!)$ such that $\mathsf{emb}(\mH) = \mathsf{emb}_K(\mH)$. 
\end{proposition}

%\begin{proof}
%We simply find an upper bound for the least common multiplier $K$ such that $K \cdot x_S$ are integers ($x_S$ are the variables in MILP~\eqref{MIP}). Following the backward direction of the previous proof, we know that for this $K$, we have $\emb_K(\mH) = \emb(\mH)$.
%The MILP has $O(2^{|V|})$ constraints where all coefficients are in $\{1, 0, -1\}$. By Cramer’s rule, a common denominator $K$ of all $x_i$ is (the absolute value of) the determinant of an $O(2^{|V|}) \times O(2^{|V|})$ matrix whose entries are all in $\{1, 0, -1\}$. Thus, $K = O((2^{|V|})!)$.
%\end{proof}

% Empirical evidence seems to suggest that a better upper bound might exist. We leave this as an open problem.
% \begin{problem}
%     Is it true that smallest $K \in \mathbb{Z}_{\geq 0}$ such that $\mathsf{emb}(H) = \mathsf{emb}_K(H)$ is polynomial in the size of $H$?
% \end{problem}

\section{Examples of Tightness}\label{sec:tight}

In this section, we identify several classes of queries where the clique embedding power coincides with the submodular width. For brevity, we write $\emb$, $\subw$, $\fhw$, and $\adw$ when the underlying hypergraph is clear under context. Table~\ref{table:summary_tightness} summarizes our results.
\begin{table}[t]
\centering
\begin{tabular}{|l||l r|l r|}
\hline
              & $\emb$       &                    & $\subw$     &                     \\ \hline
Acyclic       & $1$   &  [\autoref{tight:acyclic}]                          & $1$  &  \cite{Yannakakis81}                         \\ \hline
Chordal       & $=$   &  [\autoref{thm:chordal}]                          & $=$   & \cite{NgoPRR18}                           \\ \hline
$\ell$-cycle  & $2 - 1 / \lceil \ell / 2 \rceil$ &[\autoref{tight:cycle}] & $2 - 1 / \lceil \ell / 2 \rceil$ & \cite{DBLP:journals/algorithmica/AlonYZ97} \\ \hline
$K_{2, \ell}$ & $2 - 1/ \ell$ & [\autoref{tight:K2l}]                   & $2 - 1/ \ell$ & \cite{PANDA}                    \\ \hline
$K_{3, 3}$    & $2$  &  [\autoref{tight:K33}]                            & $2$    & \cite{PANDA}                            \\ \hline
$A_{\ell}$    & $(\ell - 1)/2$ & [\autoref{tight:Al}]                  & $(\ell - 1)/2$ & \cite{NgoPRR18}                  \\ \hline
$\mH_{\ell, k}$ & $\ell / k $ & [\autoref{tight:Hlk}]                    & $\ell / k $      & \cite{NgoPRR18}                \\ \hline
$Q_b$         & \textcolor{red}{$17/9$}    &                       & \textcolor{red}{$2$} & \cite{JoglekarR18}                                   \\ \hline
$Q_{hb}$      & \textcolor{red}{$7/4$}   &                         & \textcolor{red}{$2$} &  [\autoref{prop:subwboat}]               \\ \hline
\end{tabular}
\captionsetup{justification=centering}
\caption{Summary of $\emb$ and $\subw$ for some classes of queries}
\label{table:summary_tightness}
\end{table}
\subsection{Cycles}
 For the cycle query of length $\ell \geq 3$, we show that $\emb = \subw = 2 - 1 /\lceil \ell / 2\rceil$. The best-known algorithm for $\ell$-cycle detection (and counting) of Alon, Yuster, and Zwick~\cite{DBLP:journals/algorithmica/AlonYZ97} runs in time $O(|I|^{\subw})$. First, we show the following lemma.
 % whose proof is in~\cite{arxiv}.

\begin{lemma}\label{lem:cycle}
Consider the cycle query of length $\ell \geq 3$. Then $ \mathsf{emb} \geq  2 - 1/{\lceil \frac{\ell}{2} \rceil}$.
% \begin{enumerate}
% \item If $\ell$ is odd, $\mathsf{emb}_\ell \geq 2-2/(\ell+1)$.
% \item If $\ell$ is even, $\mathsf{emb}_{\ell-1} \geq 2-2/\ell$.
%  \end{enumerate}
\end{lemma}

\begin{proof}%[Proof of \autoref{lem:cycle}]
We start with the case where $\ell$ is odd and name the variables of the cycle query as $x_1, \dots, x_{\ell}$. Then, we  define $\lambda = (\ell+1) / 2$ and an embedding from a $\ell$-clique as follows:
\begin{equation}
    \begin{aligned}
 \psi^{-1}(x_1) & = \{ 1,2, \dots, \lambda-1\} \\ 
  \psi^{-1}(x_2) & = \{ 2, 3, \dots, \lambda\} \\ 
  & \dots \\
\psi^{-1}(x_\ell) & = \{ 2\lambda-1, 1, \dots, \lambda-2\} 
    \end{aligned}
\end{equation}
In other words, $\psi$ maps each $i \in [\ell]$ into a consecutive segment consisting of $\lambda - 1$ vertices in the cycle. To see why $\psi$ is an embedding, we observe that for any $i, j \in [\ell]$ such that $\psi(i) \cap \psi(j) = \emptyset$, $|\psi(i) \cup \psi(j)| = 2\lambda -2 = \ell - 1$, so there is an edge that intersects both $\psi(i)$ and $\psi(j)$. It is easy to see that $\mathsf{wed}(\psi) = \lambda = (\ell + 1) / 2 $. Thus, $\mathsf{emb} \geq \ell/\lambda = 2\ell/(\ell+1)= 2-2/(\ell+1)$. 

\medskip

If $\ell$ is even, we define $\lambda = \ell/ 2$ and an embedding from a $(\ell - 1)$-clique as follows:
\begin{align*}
 \psi^{-1}(x_1) & = \{ 1,2, \dots, \lambda-1\} \\ 
  \psi^{-1}(x_2) & = \{ 2, 3, \dots, \lambda\} \\ 
  & \dots \\
   \psi^{-1}(x_{\ell - 1}) & = \{ 2\lambda-2, 2\lambda-1, 1, \dots, \lambda-3\}  \\
\psi^{-1}(x_\ell) & = \psi^{-1}(x_{\ell - 1})
\end{align*}
where $\psi^{-1}(x_i), i \in [\ell - 1]$ is exactly the embedding we constructed for $(\ell - 1)$-cycle. We show that this is a valid embedding. Let $i, j \in [\ell]$ such that $\psi(i) \cap \psi(j) = \emptyset$. 
\begin{enumerate}
    \item If $i \in \psi^{-1}(x_{\ell - 1})$ (or $j \in \psi^{-1}(x_{\ell - 1})$), then $|\psi(i) \cup \psi(j)| = \ell - 1$, so there is an edge that intersects both $\psi(i)$ and $\psi(j)$.
    \item If $i, j \notin \psi^{-1}(x_{\ell - 1})$, then $\psi(i)$ and $\psi(j)$ do not contain $x_{\ell - 1}$ and $x_{\ell}$. There is an edge that intersects both $\psi(i)$ and $\psi(j)$ since $|\psi(i) \cup \psi(j)| = \ell - 2$.
\end{enumerate}
For this embedding, we have $\mathsf{wed}(\psi) = \lambda = \ell / 2 $, so $\mathsf{emb} \geq (\ell-1)/\lambda =  2-2/\ell $. 
\end{proof}

Thus, we have the following proposition.
\begin{proposition} \label{tight:cycle}
    Consider the cycle query of length $\ell \geq 3$. Then we have 
    $$ \emb = \subw = 2 - \frac{1}{\lceil \frac{\ell}{2} \rceil}
    $$
\end{proposition}
\begin{proof}
   It can be shown using Example 7.4 in~\cite{PANDA} (setting $m = 1$) that $\subw \leq 2 - 1/\lceil \frac{\ell}{2}\rceil$ (technically the Example 7.4 in~\cite{PANDA} only deals with cycles of even length, but their argument can be easily adapted to cycles of odd length). We thus conclude by applying \autoref{lem:cycle} and \autoref{prop:lowerbound}.
\end{proof}

% \hangdong{to show that $\subw \leq 2 - 1/\lceil \frac{k}{2}\rceil$ for $k$-cycle, we use Example 7.4 (setting $m = 1$) for even cycles. For odd cycles, we simply choose the vertex set as $V = I_1 \cup I_2 \cup \cdots \cup I_{2k-1}$ and the rest follows as Example 7.4 by setting $\theta = 1 - 1/k$. Together, we have shown that $\emb = \subw = 2 - 1/\lceil \frac{k}{2}\rceil$.}

% \austen{a technical point: the AYZ algorithm works for ``simple'' cycle, while the cycle query may return non-simple cycle?} \paris{The AYZ can work for simple and non-simple, but a CQ can always return repeated values.}

\subsection{Complete Bipartite Graphs} \label{sec:bipartite}
 We consider a complete bipartite graph $K_{m,n}$ where the two partitions of its vertices are $\{x_1, \dots, x_m\}$ and $\{y_1, \dots, y_n\}$. We study two of its special cases, $K_{2, \ell}$ and $K_{3, 3}$. The proofs of the following two propositions can be found in~\cite{arxiv}. 
 
\begin{proposition}\label{tight:K2l}
    For the bipartite graph $K_{2,\ell}$, $\mathsf{emb}(K_{2,\ell}) = \subw(K_{2,\ell}) =  2-1/\ell$.
\end{proposition} 

 \begin{proposition} \label{tight:K33}
     For $K_{3, 3}$, we have $\mathsf{emb}(K_{3,3}) = \subw(K_{3, 3}) =  2$.
 \end{proposition}

Finding $\mathsf{emb}(K_{m,n})$ and $\subw(K_{m, n})$ in the most general case is still an open question. 
 
 %\paris{We claim that for $K_{m,m}$, we have $\mathsf{adw} = (m+1)/2$. embedding and submodular width?}

\subsection{Chordal Queries}

In this section, we identify a special class of queries, called \textit{choral queries}. We introduce necessary definitions and lemmas to prove that for a chordal query, $\emb$, $\subw$, $\fhw$, and $\adw$ all coincide, as stated in \autoref{thm:chordal}.

Let $G$ be a graph. A chord of a cycle $C$ of $G$ is an edge  that connects two non-adjacent nodes  in $C$. We say that $G$ is {\em chordal} if any cycle in $G$ of length greater than 3 has a chord. We can extend chordality to hypergraphs by considering the clique-graph of a hypergraph $\mH$, where edges are added between all pairs of vertices contained in the same hyperedge. 

Let $(\htree, \chi)$ be a tree decomposition of a hypergraph $\mH$ and $\bags(\htree) \defeq \{\chi(t) \mid t \in V(\htree)\}$. We say that $(\htree, \chi)$ is \textit{proper} if there is no tree decomposition $(\htree', \chi')$ such that 
\begin{enumerate}
    \item for every bag $b_1 \in \bags(\htree')$, there is a bag $b_2 \in \bags(\htree)$ such that $b_1 \subseteq b_2$; 
    \item $\bags(\htree') \nsupseteq \bags(\htree)$,
\end{enumerate}

The following important properties hold for chordal graphs.

\begin{lemma}[\cite{CarmeliKKK21}]  \label{lem:maxclq}
If $G$ is a chordal graph and $(\htree, \chi)$ is a proper tree decomposition of $G$, then the bags of $(\htree, \chi)$, i.e. $\bags(\htree)$, are the maximal cliques in $G$.
\end{lemma} 

For chordal hypergraphs, we can show the following lemma:

\begin{lemma}  \label{lem:tds}
Let $\mH$ be a hypergraph. Then, $(\htree, \chi)$ is a (proper) tree decomposition of $\mH$ if and only if it is a (proper) tree decomposition of the clique-graph of $\mH$.
\end{lemma} 

\begin{proof}[Proof]
We first show that $(\htree, \chi)$ is a tree decomposition of $\mH$ if and only if it is also a tree decomposition of the clique-graph of $\mH$. The forward direction is straightforward. For the backward direction, let $(\htree, \chi)$ be a decomposition of the clique-graph of $\mH$. Then for any hyperedge $e \in \mH$ and any pair of vertices $u, v \in e$, we know that $\{t \mid u \in \chi(t)\} \cap \{t \mid v \in \chi(t)\} \neq \emptyset$. By the \textit{Helly Property}, there is a bag that contains all vertices in the hyperedge $e$. Therefore, $(\htree, \chi)$ is a tree decomposition for $\mH$. It is easy to extend the proof for proper tree decompositions.
\end{proof}

The following corollary is immediate from both \autoref{lem:maxclq} and \autoref{lem:tds}:
\begin{corollary} \label{lem:maxclq-hyper}
If $\mH$ is a chordal hypergraph and $(\htree, \chi)$ is a proper tree decomposition of $\mH$, then the bags of $(\htree, \chi)$ are the maximal cliques in the clique-graph of $\mH$.
\end{corollary} 

The above corollary tells us that every proper tree decomposition has the same set of bags, with the only difference being the way the bags are connected in the tree. From this, we can easily obtain that $\subw = \fhw$. However, we have an even stronger result:

\begin{theorem}\label{thm:chordal}
If $\mH$ is a chordal hypergprah, then $\emb = \adw = \subw = \fhw$.
\end{theorem}

\begin{proof}
Since $\mH$ is chordal, by \autoref{lem:maxclq-hyper}, the bags of any proper tree decomposition $(\htree, \chi)$ of $\mH$ are the maximal cliques in the clique-graph of $\mH$. Then, there is a node $t \in V(\htree)$ such that the minimum fractional edge cover (also the maximum fractional vertex packing) of $\chi(t)$ is $\fhw$. In particular, let $\{u_i \mid i \in \chi(t)\}$ be the optimal weights assigned to each vertex in $\chi(t)$ that obtain the maximum fractional vertex packing, so $\sum_{i \in \chi(t)} u_i = \fhw$. We let $\hat{u}_i = u_i / \sum_{i \in \chi(t)} u_i$ and $k$ be the smallest integer such that $k \cdot \hat{u}_i$ is an integer for every $i \in \chi(t)$. Now we construct an embedding $\psi$ from $C_k$ to $\mH$ so that 
every $\psi(j)$, for $j \in [k]$ is a singleton and for each $i \in \chi(t)$, let $d_{\psi}^{-1}(i) \defeq k \cdot \hat{u}_i$. This assignment uses up all $k$ vertices in $C_k$, since $\sum_{i \in \chi(t)} k \cdot \hat{u}_i = k$. Then,
$$ \mathsf{wed}(\psi) = \max_{e \in E(\mH)} \sum_{i \in e} k \cdot \hat{u}_i = \frac{k}{\sum_{i \in \chi(t)} u_i} \cdot \max_{e \in E(\mH)} \sum_{i \in e} u_i \leq \frac{k}{\sum_{i \in \chi(t)} u_i}
$$
and thus, we get $\emb = \fhw$ since
$$ \emb \geq \emb(C_k \mapsto \mH) \geq \frac{k}{\mathsf{wed}(\psi)} \geq \sum_{i \in \chi(t)} u_i = \fhw.
$$

For $\adw$, we define the following modular function over subsets of $V(\mH)$: for any $S \subseteq V(\mH)$, let $\mu(S) \defeq \sum_{i \in S} u_i$. It is edge-dominated since for every hyperedge $e$, $\mu(e) = \sum_{i \in e} u_i \leq 1$. Moreover, we have that $\mu(\chi(t)) = \sum_{i \in \chi(t)} u_i = \fhw$. That is, $\adw \geq \fhw$, so $\adw = \fhw$. As a remark, it is also viable to use Lemma 3.1 in~\cite{PANDA} to prove the claim for adaptive width.
\end{proof}

% \austen{It seems this proof can be generalized. Say a submodular width is obtained by a submodular function on bag of a tree decomposition which is a clique, then the first direction can be done (since those submodular functions are edge-dominated) and give tightness of embedding power and submodular width. For sanity check, the submodular width of the boat query is not achieved on a clique.}

Recall that Corollary~\ref{lem:maxclq-hyper} implies if $\mH$ is chordal, then every proper tree decomposition of $\mH$ has the same set of bags. We show the converse is also true, which could be of independent interest. The proof is in~\cite{arxiv}.

\begin{lemma}\label{lem:properTD}
Let $\mH$ be a hypergraph. If every proper tree decomposition of $\mH$ has the same set of bags, then $\mH$ is chordal. 
\end{lemma}

We identify three classes of hypergraphs (almost-cliques, hypercliques, and acyclic hypergraphs) that are chordal and find their clique embedding powers and submodular widths. %See~\cite{arxiv} for more discussions on the cases of almost-cliques and hypercliques. 
 
\subparagraph{Almost-cliques} \label{app:tight:chordal:Al}
 Consider the $\ell$-clique where one vertex, say $x_1$, connects to $k$ vertices only, where $1 \leq k < \ell-1$ (hence it is the missing edges from being a $\ell$-clique). We denote such a hypergraph as $A_{\ell}$. To show that $A_{\ell}$ is chordal, we observe that for any cycle of length $\geq 4$ that contains $x_1$, the two adjacent vertices of $x_1$ in the cycle must be connected by an edge in $A_{\ell}$ and that edge is a chord to the given cycle. We also show the following proposition.
 \begin{proposition}\label{tight:Al}
     For an almost-cliques $A_{\ell}$, $\emb = \subw = \fhw = (\ell-1)/2$.
 \end{proposition}
 \begin{proof}
    To prove the claim, suppose WLOG $x_i$ connects only to $x_i$, where $i \in 2, 3, \ldots, k$. Then, we take the decomposition with two bags: $\{x_1,x_2, \dots, x_{k}\}$, $\{x_2, x_3, \dots, x_\ell\}$, where each bag has an edge cover of at most $(\ell-1)/2$ since the first bag induces an $k$-clique and the second bag induces an $(\ell-1)$-clique. Hence, $\fhw \leq (\ell-1)/2$.

On the other hand, consider the embedding $\psi$ from the $(\ell-1)$-clique, where $\psi(i) = x_{i+1}$, $1 \leq i \leq \ell-1$; it is easy to verify that this is a valid embedding such that $\mathsf{wed}(\psi) = 2$, hence $\mathsf{emb} \geq (\ell-1)/2$. Therefore, we have shown that $\mathsf{emb} = \subw = (\ell - 1)/2$ by \autoref{prop:lowerbound}. 
\end{proof}

\subparagraph{Hypercliques} \label{app:tight:chordal:Hlk}
Next, we consider the $(\ell,k)$-hyperclique query $\mathcal{H}_{\ell,k}$, where $1 < k \leq \ell$. This query has $\ell$ variables, and includes as atoms all possible subsets of $\{x_1, \dots, x_\ell\}$ of size exactly $k$. When $k=\ell-1$, the query simply becomes a Loomis-Whitney join~\cite{NgoPRR18}. It is easy to see that $\mathcal{H}_{\ell,k}$ is chordal since the clique-graph of $\mathcal{H}_{\ell,k}$ is a $\ell$-clique.

\begin{proposition} \label{tight:Hlk}
For $\mathcal{H}_{\ell,k}$, we have $\mathsf{emb} = \subw = \fhw = \ell/k$.
\end{proposition}

\begin{proof}
First, we show that $\fhw \leq \ell/k$. Indeed, there is a fractional edge cover that assigns a weight of $1/k$ to each hyperedge that contains $k$ consecutive vertices in $\{x_1, \dots, x_\ell\}$ (let the successor of $x_{\ell}$ be $x_{1}$). The fractional edge cover is then $\ell/k$. We show next that this bound coincides with $\mathsf{emb}$.

We simply define the embedding $\psi$ from a $\ell$-clique as $\psi(i) = x_i, i\in [\ell]$. Then, $\mathsf{wed}(\psi) = k$ since every hyperedge has exactly $k$ vertices. Therefore, $\mathsf{emb} \geq \ell / k \geq \fhw$ and we conclude by applying \autoref{lem:cycle} and \autoref{prop:lowerbound}.
\end{proof}

 \subparagraph{Acyclic Hypergraphs}
First, we claim that acyclic queries are indeed chordal queries. %The proof can be found in~\cite{arxiv}.
\begin{lemma}\label{lem:acyclic}
    An acyclic hypergraph $\mH$ is chordal.
\end{lemma}
\begin{proof}
 We prove by induction on the number of hyperedges in the hypergraph $\mH = (V, E)$. If $|E| = 1$, it is clique-graph is a clique, thus it is chordal. The induction hypothesis assumes the claim for acyclic hypergraphs with $|E| \leq k$ hyperedges. Let $\mH$ be an acyclic hypergraph such that $|E| = k+ 1$. Since $\mH$ is acyclic, it has a join forest whose vertices are the hyperedges of $\mH$. Let $e_{\ell} \in E$ be a leaf of the join forest and it is easy to show that $\mH' = (V, E \setminus \{e_{\ell}\})$ is an acyclic hypergraph with $k$ hyperedges. For any cycle in the clique-graph of $\mH$ having length $\geq 4$, we discuss the following two cases.

 If every edge of the cycle is in the clique-graph of $\mH'$: by the induction hypothesis, there is a chord for this cycle in the clique-graph of $\mH'$ (thus also in $\mH$).

 Otherwise, there is an edge $e$ in the cycle that is in the clique-graph of $\mH$, but not in the clique-graph of $\mH'$: therefore, the edge $e$ is only contained by $e_{\ell}$. This implies that there is a vertex $u$ that is only contained by $e_{\ell}$, not by any other edges in $E$. Let $\{u, v\}$ and $\{u, w\}$ be the edges connecting $u$ in the given cycle, we know that $\{u, v, w\} \subseteq e_{\ell}$ and thus, $\{v, w\}$ is a chord for the given cycle.
\end{proof}

 Now we prove the following theorem for acyclic hypergraphs:
\begin{theorem} \label{tight:acyclic}
 For an acyclic hypergprah $\mH$, $\emb = \adw = \subw = \fhw = 1$.
\end{theorem}
\begin{proof}
    From \autoref{prop:emb:properties}, we know that $\emb \geq 1$. Since it is known that $\subw = \fhw = 1$, the theorem is then a direct result from \autoref{lem:acyclic} and \autoref{thm:chordal}.
\end{proof}
 
\section{Gap Between Clique Embedding Power and Submodular Width}\label{sec:boat}

In this section, we discuss the boat query and its variant depicted in Figure~\ref{fig: boat}, where gaps between the clique embedding power and submodular width can be shown.

%Indeed, consider the embedding of the 5-clique where we map $\{1,2\}$ to $\{x_1, x_2, x_3\}$ and $3,4,5$ to $y_1, y_2, y_3$ respectively. This embedding has weak edge depth 3. \paris{Is this optimal?}

\begin{figure}[h]
\centering
\begin{subfigure}{0.4\textwidth}
\centering
\begin{tikzpicture}
    \node[shape=circle] (x1) at (-1.5,0) {$x_1$};
    \node[shape=circle] (x4) at (-0.5,0) {$x_4$};
    \node[shape=circle] (x5) at (0.5,0) {$x_5$};
    \node[shape=circle] (x8) at (1.5,0) {$x_8$};
    \node[shape=circle] (x2) at (-0.5,1) {$x_2$};
    \node[shape=circle] (x3) at (0.5,1) {$x_3$};
    \node[shape=circle] (x6) at (-0.5,-1) {$x_6$};
    \node[shape=circle] (x7) at (0.5,-1) {$x_7$};
    
    \draw (x1) -- (x2) ;
    \draw (x1) -- (x4) ;
    \draw (x1) -- (x6) ;
    \draw (x2) -- (x3) ;
    \draw (x4) -- (x5) ;
    \draw (x6) -- (x7) ;
    \draw (x3) -- (x8) ;
    \draw (x5) -- (x8) ;
    \draw (x7) -- (x8) ;
\end{tikzpicture}
\captionsetup{justification=centering}
\caption{Boat Query $Q_b$}
\end{subfigure}
\begin{subfigure}{0.4\textwidth} 
\centering
\begin{tikzpicture}
    \node[shape=circle] (y1) at (-1,1) {$y_1$};
    \node[shape=circle] (y2) at (-1,0) {$y_2$};
    \node[shape=circle] (y3) at (-1,-1) {$y_3$};
    \node[shape=circle] (z1) at (1,1) {$z_1$};
    \node[shape=circle] (z2) at (1,0) {$z_2$};
    \node[shape=circle] (z3) at (1,-1) {$z_3$};
    
    \draw (z1) -- (y1) ;
    \draw (z2) -- (y2) ;
    \draw (z3) -- (y3) ;
    
    \node (box1) [draw=black, inner sep=1mm, rectangle, rounded corners,  fit = (y1) (y3)] {};
    \node (box2) [draw=black, inner sep=1mm, rectangle, rounded corners,  fit = (z1) (z3)] {};
\end{tikzpicture}
\captionsetup{justification=centering}
\caption{Hyper-boat Query $Q_{hb}$}
\end{subfigure}
\captionsetup{justification=centering}
\caption{The Boat query and its variant, the Hyper-boat Query}
\label{fig: boat}

%% Middle Boat Query
% \begin{subfigure}{0.3\textwidth}
% \centering
% \begin{tikzpicture}
%     \node[shape=circle] (w1) at (-0.5,1) {$w_1$};
%     \node[shape=circle] (w2) at (-0.5,0) {$w_2$};
%     \node[shape=circle] (w3) at (-0.5,-1) {$w_3$};
%     \node[shape=circle] (w4) at (1,1) {$w_4$};
%     \node[shape=circle] (w5) at (1,0) {$w_5$};
%     \node[shape=circle] (w6) at (1,-1) {$w_6$};
%     \node[shape=circle] (w7) at (2.5,0) {$w_7$};

%     \draw (w1) -- (w4) ;
%     \draw (w2) -- (w5) ;
%     \draw (w3) -- (w6) ;
%     \draw (w4) -- (w7) ;
%     \draw (w5) -- (w7) ;
%     \draw (w6) -- (w7) ;
    
%     \node (box1) [draw=black, inner sep=1mm, rectangle, rounded corners,  fit = (w1) (w3)] {};
%     %\node (box2) [draw=white, inner sep=1mm, rectangle, rounded corners,  fit = (z1) (z3)] {};
    
% \end{tikzpicture}
% \captionsetup{justification=centering}
% \caption{Middle Boat Query $Q_{mb}$}
% \end{subfigure}
\end{figure}

%% Middle Boat Query
% \begin{subfigure}{0.3\textwidth}
% \centering
% \begin{tikzpicture}
%     \node[shape=circle] (w1) at (-0.5,1) {$w_1$};
%     \node[shape=circle] (w2) at (-0.5,0) {$w_2$};
%     \node[shape=circle] (w3) at (-0.5,-1) {$w_3$};
%     \node[shape=circle] (w4) at (1,1) {$w_4$};
%     \node[shape=circle] (w5) at (1,0) {$w_5$};
%     \node[shape=circle] (w6) at (1,-1) {$w_6$};
%     \node[shape=circle] (w7) at (2.5,0) {$w_7$};

%     \draw (w1) -- (w4) ;
%     \draw (w2) -- (w5) ;
%     \draw (w3) -- (w6) ;
%     \draw (w4) -- (w7) ;
%     \draw (w5) -- (w7) ;
%     \draw (w6) -- (w7) ;
    
%     \node (box1) [draw=black, inner sep=1mm, rectangle, rounded corners,  fit = (w1) (w3)] {};
%     %\node (box2) [draw=white, inner sep=1mm, rectangle, rounded corners,  fit = (z1) (z3)] {};
    
% \end{tikzpicture}
% \captionsetup{justification=centering}
% \caption{Middle Boat Query $Q_{mb}$}
% \end{subfigure}

\subsection{Clique Embedding Power, Submodular Width and Adaptive Width} \label{sec:gaps}
\begin{figure}[h]
\centering
\begin{subfigure}{0.47\textwidth}
\centering
\begin{tikzpicture}
    \node[shape=circle] (x1) at (-3,0) {1,6-11,16};
    \node[shape=circle] (x4) at (-1,0) {5-11};
    \node[shape=circle] (x5) at (1,0) {3-5};
    \node[shape=circle] (x8) at (3,0) {2-4,12-17};
    \node[shape=circle] (x2) at (-1,1) {1,2,6-8};
    \node[shape=circle] (x3) at (1,1) {1,2,12-15};
    \node[shape=circle] (x6) at (-1,-1) {9-11,16-17};
    \node[shape=circle] (x7) at (1,-1) {12-17};
    \draw (x1) -- (x2) ;
    \draw (x1) -- (x4) ;
    \draw (x1) -- (x6) ;
    \draw (x2) -- (x3) ;
    \draw (x4) -- (x5) ;
    \draw (x6) -- (x7) ;
    \draw (x3) -- (x8) ;
    \draw (x5) -- (x8) ;
    \draw (x7) -- (x8) ;
\end{tikzpicture}
\captionsetup{justification=centering}
\caption{$\mathsf{emb}$ for $Q_{b}$}
\end{subfigure}
\begin{subfigure}{0.47\textwidth}
\centering
\begin{tikzpicture}
    \node[shape=circle] (y1) at (0,0) {1-3};
    \node[shape=circle] (y2) at (0,1) {2,3};
    \node[shape=circle] (y3) at (0,2) {1,4};
    \node[shape=circle] (z1) at (2,0) {7};
    \node[shape=circle] (z2) at (2,1) {5,6};
    \node[shape=circle] (z3) at (2,2) {4-6};
    
    \draw (y1) -- (z1) ;
    \draw (y2) -- (z2) ;
    \draw (y3) -- (z3) ;
    
    \node (box1) [draw=black, inner sep=1mm, rectangle, rounded corners,  fit = (y1) (y3)] {};
    \node (box2) [draw=black, inner sep=1mm, rectangle, rounded corners,  fit = (z1) (z3)] {};
\end{tikzpicture}
\captionsetup{justification=centering}
\caption{$\mathsf{emb}$ for $Q_{hb}$}
\end{subfigure}
% \begin{subfigure}{0.45\textwidth}
% \centering
% \begin{tikzpicture}
%     \node[shape=circle] (w1) at (-1,1) {1-3};
%     \node[shape=circle] (w2) at (-1,0) {4,5,10};
%     \node[shape=circle] (w3) at (-1,-1) {2-5,13};
%     \node[shape=circle] (w4) at (1,1) {1,6-9};
%     \node[shape=circle] (w5) at (1,0) {6-10};
%     \node[shape=circle] (w6) at (1,-1) {11-13};
%     \node[shape=circle] (w7) at (3,0) {6-9,11,12};

%     \draw (w1) -- (w4) ;
%     \draw (w2) -- (w5) ;
%     \draw (w3) -- (w6) ;
%     \draw (w4) -- (w7) ;
%     \draw (w5) -- (w7) ;
%     \draw (w6) -- (w7) ;
    
%     \node (box1) [draw=black, inner sep=1mm, rectangle, rounded corners,  fit = (w1) (w3)] {};
%     %\node (box2) [draw=white, inner sep=1mm, rectangle, rounded corners,  fit = (z1) (z3)] {};
    
% \end{tikzpicture}
% \captionsetup{justification=centering}
% \caption{$\mathsf{emb}$ for $Q_{mb}$}
% \end{subfigure}
\captionsetup{justification=centering}
\caption{Optimal embedding for the boat query and its variants}
\label{fig: embedding for boat}
\end{figure}

Using MILP~\eqref{MIP}, we find the optimal clique embedding for $Q_{b}$ and $Q_{hb}$, as illustrated in Figure~\ref{fig: embedding for boat}. The numbers represent the vertices from the clique, and we adopt the shorthand notation, say, 6-8 to refer to the set \{6,7,8\}. The clique embedding powers for $Q_b$ and $Q_{hb}$ are $\frac{17}{9}$ and $\frac{7}{4}$, respectively. However, \cite{JoglekarR18} proves that for the boat query, $\subw(Q_b) = 2$. This implies a gap since $\emb(Q_b) = 17 /9 < \subw(Q_b) = 2$.
We show that for the hyper-boat query $Q_{hb}$, there is also a gap between the optimal clique embedding power and submodular width. In particular, we show that $\subw(Q_{hb}) = 2$ in the following proposition, which implies the following gap: $\emb(Q_{hb}) = \frac{7}{4} < \subw(Q_{hb}) =2 $. Its proof can be found in~\cite{arxiv}.

\begin{proposition} \label{prop:subwboat}
 For $Q_{hb}$, we have $\subw(Q_{hb}) = 2$. 
\end{proposition}

% \hangdong{not finished...}
% \begin{proposition}
%  For $Q_{hb}$, we have $\adw(Q_{hb}) = 3/2$. 
% \end{proposition}
%   First, we show that $\adw(Q_{hb}) \geq 3/2$ by defining a modular function $\mu$ as follows: $\mu(z_1) = \mu(y_1) = \mu(z_2) = \mu(y_2) = 1/2$, and $\mu(z_3) = \mu(z_3) = 0$. Any decomposition needs to have a bag that contains at least three of $\{x_1, y_1, x_2, y_2\}$, and this bag will have width $3/2$.

\subsection{Subquadratic Equivalence Between Boat Queries}\label{sec:subquadratic}

In this section, we demonstrate an interesting connection between the two boat queries. To start, let's consider $Q_b$ and $Q_{hb}$. Both queries admit an algorithm that runs in time $O(|I|^2)$. Informally, we are going to show that either both queries can be executed significantly faster, or neither can. Following the seminal paper by Williams and Williams~\cite{subcubic}, we define {\em truly subquadratic algorithm} and {\em subquadratic equivalence}. 
\begin{definition}
An algorithm is said to be \emph{truly subquadratic} if it runs in time $O(m^{2-\epsilon})$ for some $\epsilon > 0$ ($m$ is the input size).
\end{definition}

Two problems $A$ and $B$ are \emph{subquadratic equivalent} if $A$ admits a truly subquadratic algorithm iff $B$ admits a truly subquadratic algorithm.
We show that the two boat queries are subquadratic equivalent. %The proof of the following theorem can be found in~\cite{arxiv}. 

\begin{theorem} \label{thm:subquadratic}
$Q_b$ is subquadratic equivalent to $Q_{hb}$.
\end{theorem}

\begin{proof}
It's easy to see that a truly subquadratic algorithm for $Q_{b}$ gives a truly subquadratic algorithm for $Q_{hb}$. Indeed, given an input instance $I_{hb}$ of $Q_{hb}$, we can form an input instance $I_b$ of $Q_b$ where the table $(x_1,x_2)$ is the projection of the table $(y_1,y_2,y_3)$ in $I_{hb}$, and similar for the tables $(x_1,x_4), (x_1,x_6), (x_3,x_8), (x_5,x_8)$ and $(x_7, x_8)$. We then solve $I_b$ by the algorithm for $Q_b$. It is easy to see that this algorithm is correct and runs in truly quadratic time.

The converse direction needs more work, since if we were to simply create the table $(y_1, y_2, y_3)$ for $Q_{hb}$ by joining the tables $(x_1, x_2), (x_1,x_4)$ and $ (x_1, x_6)$ for $Q_b$, the size of the result might be significantly greater than all previous tables. For example, if the sizes of the tables $(x_1, x_2), (x_1,x_4)$ and $ (x_1, x_6)$ are all $m$, then joining them could result in a table of size $m^{\frac{3}{2}}$ and therefore calling the algorithm for $Q_{hb}$ on this instance does not necessarily yield a truly subquadratic algorithm for $Q_{b}$.

We perform our fine-grained reduction based on \emph{heavy-light split}. Our goal is to give a subquadratic algorithm for $Q_b$ assuming there is one such algorithm for $Q_{hb}$. Suppose the subquadratic algorithm for $Q_b$ runs in time $O(m^{2-\delta})$ for some $\delta > 0$, where $m$ is the size of all tables. Our algorithm for $Q_b$ runs as follows. First, it checks whether there are entries of attribute $x_1$ that has degree more than $\Delta := m^\epsilon$ in tables $(x_1,x_2), (x_1,x_4)$ and $(x_1,x_6)$ for some $\epsilon > 0$ to be specified later. Those are called \emph{heavy} and there are at most $\frac{m}{\Delta}$ many of them. For those entries, we fix each one so that the remaining query becomes acyclic, and thus can be solved in linear time by Yannakakis algorithm~\cite{Yannakakis81}. We do the same procedure for heavy entries of attribute $x_8$. Therefore, any result of $Q_{b}$ that contains a heavy entry in attributes $x_1$ or $x_8$ will be detected in time $O(m^{2-\epsilon})$. It remains to consider the case where the entries of attributes $x_1$ and $x_8$ have degrees less than $\Delta$, which are called \emph {light}. In this case, we loop over all light entries of $x_1$ in the table $(x_1,x_2)$ and directly join them with the tables $(x_1,x_4)$ and $(x_1,x_6)$ and project the result to build a table $(x_2,x_4,x_6)$. We then do the same procedure for joining $x_8$. This will cost time $O(m\cdot \Delta \cdot \Delta) = O(m^{1+2\epsilon})$. We then call the $O(m^{2-\delta})$ algorithm for $Q_{hb}$, which cost time $O(m^{(1+2\epsilon)(2-\delta)})$. By choosing $0 < \epsilon < \frac{\delta}{4-2\delta}$ (note that $\delta < 2$), we observe that the whole algorithm for $Q_b$ takes time $O(m^{2-\epsilon}) + O(m^{(1+2\epsilon)(2-\delta)}) = O(m^{2-\epsilon'})$ for some $\epsilon' > 0$.
\end{proof}

We remark that the reduction from $Q_{hb}$ to $Q_b$ is parametrized by the running time of the algorithm for  $Q_{hb}$. That is, the reduction is not uniform in the sense that only after given $\delta>0$ can we specify a suitable $\epsilon$. Theorem~\ref{thm:subquadratic} implies that either both boat queries admit a truly subquadratic algorithm or none of them does.

The fact that there is a gap between $\subw(Q_{hb}) = 2$ and $\mathsf{emb}(Q_{hb}) = \frac{7}{4}$ suggests currently our lower bound does not match with the best upper bound, i.e., $\PANDA$. This implies either that $\PANDA$ is not universally optimal, or that we are missing the best possible lower bound.  We leave this as an open question. 

Finally, we note that Theorem 4 in~\cite{JoglekarR18}, which proves there does not exist a $\tilde{O}(m^{2-\epsilon}+|\text{OUT}|\footnote{$|\text{OUT}|$ is the size of the output.})$ algorithm for the boat query unless 3-XOR can be solved in time $\tilde{O}(m^{2-t})$ for a $t>0$, does not directly translate into the quadratic hardness for the boat query in our case. This is because their reduction uses the output of the boat query in an essential way to ``hack back the collision'' which is not available in the Boolean case.
\section{Related Work}

\subparagraph*{Fine-Grained Complexity} The study of fine-grained complexity aims to show the (conditional) hardness of easy problems. Recent years have witnessed a bloom of development into this fascinating subject, resulting in many tight lower bounds which match exactly, or up to poly log factors, the running time of best-known algorithms~\cite{kclique, monochromatic,subcubic,ValiantParser}. Among many others, popular hardness assumptions include the Strong Exponential Time Hypothesis (SETH), Boolean Matrix Multiplication (BMM), and All-Pairs Shortest Paths (APSP). Our work can be seen as a particular instance under this framework, i.e. using Boolean or Min-Weight $k$-Clique Conjecture to show conditional lower bounds for BCQs. Interestingly, our reduction of $k$-cycles essentially mirrors the construction in the proof of Theorem 3.1 in~\cite{kclique}.

\subparagraph*{Conjunctive Queries (CQs) Evaluation} The efficient evaluation of CQs constitutes the core theme of database theory. Khamis, Ngo, and Suciu introduced in~\cite{PANDA} the $\mathsf{PANDA}$ algorithm that runs in time as predicted by the submodular width of the query hypergraph. This groundbreaking result establishes a profound connection between various lines of work on tree decompositions \cite{Marx10, Marx13}, worst-case optimal join algorithms \cite{NgoRR13, NgoPRR18}, and the interplay between CQ evaluation and information theory \cite{KhamisK0S20, zhao2023space, FAQAI}.

\subparagraph*{Functional Aggregate Queries (FAQs)} FAQs~\cite{FAQ} provides a Sum-of-Product framework that captures the semantics of conjunctive queries over arbitrary semirings. The semiring point-of-view originated from the seminal paper~\cite{GreenKT07}. Khamis, Ngo, and Rudra \cite{FAQ} initiate the study of the efficient evaluation of FAQs. \cite{FAQAI} introduces the FAQ version of the submodular width $\#\subw$ and the $\#\textsf{PANDA}$ algorithm (as the FAQ version of the $\mathsf{PANDA}$ algorithm) that achieves the runtime as predicated by $\#\subw$. We show that the embedding from a $k$-clique into a hypergraph holds for arbitrary semirings, which enables one to transfer the hardness of $k$-clique to FAQ independent of the underlying semiring. To the best of our knowledge, this is the first \emph{semiring-oblivious} reduction. 

\subparagraph*{Enumeration and Preprocessing} Bagan, Durand and Grandjean characterized in \cite{BaganDG07} when a constant delay and linear preprocessing algorithm for self-join-free conjunctive queries  is possible. 
A recent paper~\cite{CarmeliS22} makes an initial foray towards the characterization of conjunctive queries with self-joins. Also recently, \cite{Nofar22} identifies new queries which can be solved with linear preprocessing time and constant delay. Their hardness results are based on the Hyperclique conjecture, the Boolean Matrix Multiplication conjecture, and the 3SUM conjecture.
\section{Conclusion}

In this paper, we study the fine-grained complexity of BCQs. We give a semiring-oblivious reduction from the $k$-clique problem to an arbitrary hypergraph. Assuming the Boolean $k$-Clique Conjecture, we obtain conditional lower bounds for many queries that match the combinatorial upper bound achieved by the best-known algorithms, possibly up to a poly-logarithmic factor. 

One attractive future direction is to fully unravel the gap between the clique embedding power and submodular width, where improved lower bounds or upper bounds are possible.
The Boolean $k$-Clique Conjecture states that there is no $O(n^{k-\epsilon})$ \emph{combinatorial} algorithm for detecting $k$-cliques. One future direction is to base the hardness assumption over Ne{\v{s}}et{\v{r}}il and Poljak's algorithm~\cite{nevsetvril1985complexity}, which solves the $k$-clique problem in $O\left(n^{(\omega / 3) k}\right)$ by leveraging fast matrix multiplication techniques and show lower bounds for any algorithm. 

\bibliography{ref}

\begin{thebibliography}{10}

\bibitem{ValiantParser}
Amir Abboud, Arturs Backurs, and Virginia~Vassilevska Williams.
\newblock If the current clique algorithms are optimal, so is {V}aliant's
  parser.
\newblock In {\em {FOCS}}, pages 98--117. {IEEE} Computer Society, 2015.

\bibitem{DBLP:journals/algorithmica/AlonYZ97}
Noga Alon, Raphael Yuster, and Uri Zwick.
\newblock Finding and counting given length cycles.
\newblock {\em Algorithmica}, 17(3):209--223, 1997.

\bibitem{BaganDG07}
Guillaume Bagan, Arnaud Durand, and Etienne Grandjean.
\newblock On acyclic conjunctive queries and constant delay enumeration.
\newblock In {\em {CSL}}, volume 4646 of {\em Lecture Notes in Computer
  Science}, pages 208--222. Springer, 2007.

\bibitem{Nofar22}
Karl Bringmann and Nofar Carmeli.
\newblock Unbalanced triangle detection and enumeration hardness for unions of
  conjunctive queries.
\newblock {\em CoRR}, abs/2210.11996, 2022.

\bibitem{CarmeliKKK21}
Nofar Carmeli, Batya Kenig, Benny Kimelfeld, and Markus Kr{\"{o}}ll.
\newblock Efficiently enumerating minimal triangulations.
\newblock {\em Discret. Appl. Math.}, 303:216--236, 2021.

\bibitem{CarmeliS22}
Nofar Carmeli and Luc Segoufin.
\newblock Conjunctive queries with self-joins, towards a fine-grained
  complexity analysis.
\newblock {\em CoRR}, abs/2206.04988, 2022.

\bibitem{CaselS21}
Katrin Casel and Markus~L. Schmid.
\newblock Fine-grained complexity of regular path queries.
\newblock In {\em {ICDT}}, volume 186 of {\em LIPIcs}. Schloss Dagstuhl -
  Leibniz-Zentrum f{\"{u}}r Informatik, 2021.

\bibitem{survey}
Arnaud Durand.
\newblock Fine-grained complexity analysis of queries: From decision to
  counting and enumeration.
\newblock In {\em {PODS}}, pages 331--346. {ACM}, 2020.

\bibitem{arxiv}
Austen~Z. Fan, Paraschos Koutris, and Hangdong Zhao.
\newblock The fine-grained complexity of boolean conjunctive queries and
  sum-product problems.
\newblock 2023.
\newblock \href {http://arxiv.org/abs/2304.14557} {\path{arXiv:2304.14557}}.

\bibitem{GreenKT07}
Todd~J. Green, Gregory Karvounarakis, and Val Tannen.
\newblock Provenance semirings.
\newblock In {\em {PODS}}, pages 31--40. {ACM}, 2007.

\bibitem{grohe2007complexity}
Martin Grohe.
\newblock The complexity of homomorphism and constraint satisfaction problems
  seen from the other side.
\newblock {\em Journal of the ACM (JACM)}, 54(1):1--24, 2007.

\bibitem{Heggernes05}
Pinar Heggernes.
\newblock Treewidth, partial $k$-trees, and chordal graphs.
\newblock {\em Partial curriculum in INF334-Advanced algorithmical techniques,
  Department of Informatics, University of Bergen, Norway}, 2005.

\bibitem{JoglekarR18}
Manas Joglekar and Christopher R{\'{e}}.
\newblock It's all a matter of degree - using degree information to optimize
  multiway joins.
\newblock {\em Theory Comput. Syst.}, 62(4):810--853, 2018.

\bibitem{FAQAI}
Mahmoud~Abo Khamis, Ryan~R. Curtin, Benjamin Moseley, Hung~Q. Ngo, XuanLong
  Nguyen, Dan Olteanu, and Maximilian Schleich.
\newblock On functional aggregate queries with additive inequalities.
\newblock In {\em {PODS}}, pages 414--431. {ACM}, 2019.

\bibitem{KhamisK0S20}
Mahmoud~Abo Khamis, Phokion~G. Kolaitis, Hung~Q. Ngo, and Dan Suciu.
\newblock Bag query containment and information theory.
\newblock In {\em {PODS}}, pages 95--112. {ACM}, 2020.

\bibitem{FAQ}
Mahmoud~Abo Khamis, Hung~Q. Ngo, and Atri Rudra.
\newblock {FAQ:} questions asked frequently.
\newblock In {\em {PODS}}, pages 13--28. {ACM}, 2016.

\bibitem{PANDA}
Mahmoud~Abo Khamis, Hung~Q. Ngo, and Dan Suciu.
\newblock What do shannon-type inequalities, submodular width, and disjunctive
  datalog have to do with one another?
\newblock In {\em {PODS}}, pages 429--444. {ACM}, 2017.

\bibitem{kclique}
Andrea Lincoln, Virginia~Vassilevska Williams, and R.~Ryan Williams.
\newblock Tight hardness for shortest cycles and paths in sparse graphs.
\newblock In {\em {SODA}}, pages 1236--1252. {SIAM}, 2018.

\bibitem{Marx10}
D{\'{a}}niel Marx.
\newblock Can you beat treewidth?
\newblock {\em Theory Comput.}, 6(1):85--112, 2010.

\bibitem{Marx13}
D{\'{a}}niel Marx.
\newblock Tractable hypergraph properties for constraint satisfaction and
  conjunctive queries.
\newblock {\em J. {ACM}}, 60(6):42:1--42:51, 2013.

\bibitem{nevsetvril1985complexity}
Jaroslav Ne{\v{s}}et{\v{r}}il and Svatopluk Poljak.
\newblock On the complexity of the subgraph problem.
\newblock {\em Commentationes Mathematicae Universitatis Carolinae},
  26(2):415--419, 1985.

\bibitem{NgoPRR18}
Hung~Q. Ngo, Ely Porat, Christopher R{\'{e}}, and Atri Rudra.
\newblock Worst-case optimal join algorithms.
\newblock {\em J. {ACM}}, 65(3):16:1--16:40, 2018.

\bibitem{NgoRR13}
Hung~Q. Ngo, Christopher R{\'{e}}, and Atri Rudra.
\newblock Skew strikes back: new developments in the theory of join algorithms.
\newblock {\em {SIGMOD} Rec.}, 42(4):5--16, 2013.

\bibitem{RoseTL76}
Donald~J. Rose, Robert~Endre Tarjan, and George~S. Lueker.
\newblock Algorithmic aspects of vertex elimination on graphs.
\newblock {\em {SIAM} J. Comput.}, 5(2):266--283, 1976.

\bibitem{Schrijver99}
Alexander Schrijver.
\newblock {\em Theory of linear and integer programming}.
\newblock Wiley-Interscience series in discrete mathematics and optimization.
  Wiley, 1999.

\bibitem{subcubic}
Virginia~Vassilevska Williams and R.~Ryan Williams.
\newblock Subcubic equivalences between path, matrix, and triangle problems.
\newblock {\em J. {ACM}}, 65(5):27:1--27:38, 2018.

\bibitem{monochromatic}
Virginia~Vassilevska Williams and Yinzhan Xu.
\newblock Monochromatic triangles, triangle listing and {APSP}.
\newblock In {\em {FOCS}}, pages 786--797. {IEEE}, 2020.

\bibitem{Yannakakis81}
Mihalis Yannakakis.
\newblock Algorithms for acyclic database schemes.
\newblock In {\em {VLDB}}, pages 82--94. {IEEE} Computer Society, 1981.

\bibitem{zhao2023space}
Hangdong Zhao, Shaleen Deep, and Paraschos Koutris.
\newblock Space-time tradeoffs for conjunctive queries with access patterns.
\newblock 2023.
\newblock \href {http://arxiv.org/abs/2304.06221} {\path{arXiv:2304.06221}}.

\end{thebibliography}

\newpage
\appendix
\section{Missing Details from \autoref{sec:emb}} \label{app:emb}

\begin{proof}[Proof of \autoref{thm:adw}]
Let $\mathsf{ed}(C_k \mapsto \mH)=\alpha$. Then, there is an embedding $\psi: C_k \mapsto \mH$ with edge depth $\alpha$. We will show that $\mathsf{adw}(\mH) \geq k/\alpha$.

First, we define the following function over subsets of $V(\mH)$: for any $S \subseteq V(\mH)$, let $\mu(S) = \sum_{v \in S} d_\psi(v)/\alpha$. This function forms a fractional independent set, since  for any hyperedge $e$, we have $\mu(e) = \sum_{v \in e } d_\psi(v)/\alpha = d^+_\psi(e)/\alpha \leq 1$.

Now, consider any decomposition $(\htree, \chi)$ of $\mH$. From Lemma~\ref{lem:decomp}, there is a node $t \in T$ such that or every $i=1, \dots, k$, $\psi(i) \cap \chi(t) \neq \emptyset$. Hence, $\mu(B_t) =  \sum_{v \in B_t} d_\psi(v) /\alpha \geq k/\alpha$. Thus, the adaptive width of the decomposition is at least $k/\alpha$.
\end{proof}

\section{Missing Details from \autoref{sec: decidability}} \label{app:decidable}

\begin{proof}[Proof of \autoref{thm: fractional emb}]
Proposition~\ref{prop:milp} shows that to compute ${\emb}(\mH)$ it suffices to solve MILP~\eqref{MIP}. To solve the MILP, we can sequentially fix the assignments of the binary variables $y_S$ and then solve a linear program. Note that the remaining linear program might have exponentially many conditions in the size of $\mH$, since the number of variables is $2^{|V|}$. The proof is then completed by observing the number of assignments of $y_S$ are doubly exponential in $\mH$.
\end{proof}

\begin{proof}[Proof of \autoref{prop:Kbound}]
   We simply find an upper bound for the least common multiplier $K$ such that $K \cdot x_S$ are integers ($x_S$ are the variables in MILP~\eqref{MIP}). Following the backward direction of the previous proof, we know that for this $K$, we have $\emb_K(\mH) = \emb(\mH)$.
The MILP has $O(2^{|V|})$ constraints where all coefficients are in $\{1, 0, -1\}$. By Cramer’s rule, a common denominator $K$ of all $x_i$ is (the absolute value of) the determinant of an $O(2^{|V|}) \times O(2^{|V|})$ matrix whose entries are all in $\{1, 0, -1\}$. Thus, $K = O((2^{|V|})!)$.
\end{proof}

\section{Missing Details from \autoref{sec:tight}} \label{app:tight}

\subsection{Complete Bipartite Graphs} \label{app:tight:bipartite}
\begin{proof}[Proof of \autoref{tight:K2l}]
We define the embedding $\psi$ from a $(2\ell - 1)$-clique as follows: $\psi^{-1}(x_1) = \{1, \dots, \ell-1\}$,  $\psi^{-1}(x_2) = \{\ell, \dots, 2\ell-2\}$, and $\psi^{-1}(y_i) = \{i, \ell+i-1\}$ for $1 \leq i \leq \ell-1$, while $\psi^{-1}(y_\ell) = \{2\ell-1\}$. To show that $\psi$ is an embedding, we observe that $\psi(2\ell - 1) = \{y_{\ell}\}$ and for each $i \in [2\ell - 2]$, $\psi(i)$ contains exactly two vertices, one from $\{x_1, x_2\}$ and the other from $\{y_1, \dots, y_n\}$, so $\psi(i)$ induces an edge as a subgraph. Thus, for any $i, j \in [2\ell - 1]$, $\psi(i)$ and $\psi(j)$ touch because there is an edge $(u, v)$, where $u \in \psi(i) \cap \{x_1, x_2\}, v \in \psi(j) \cap \{y_1, \dots, y_n\}$, that intersects both $\psi(i)$ and $\psi(j)$. It is easy to see that $\mathsf{wed}(\psi) = \ell$, thus $\emb_{2 \ell - 1} \geq 2 - 1/\ell$.

Next, we show that $\subw(K_{2,\ell}) \leq 2 - 1/\ell$. There are only two proper tree decompositions for $K_{2,\ell}$, where the first one has $\ell$ bags $\{x_1,x_2, y_1\}, \dots, \{x_1, x_2, y_\ell\}$ and the second one has two bags $\{x_1, y_1, \dots, y_\ell\}, \{x_2, y_1, \dots, y_\ell\}$. Let $h$ be any edge-dominated submodular function. 
\begin{enumerate}
    \item [(1)] $h(x_i) \leq \theta$ for some $i \in \{1, 2\}$. WLOG we assume $h(\{x_1\}) \leq \theta$. Then, for any bag in the first decomposition, i.e., $\{x_1, x_2, y_i\}$, $i \in [\ell]$, we have 
    $$ h(\{x_1, x_2, y_i\}) \leq h(\{x_1\}) + h(\{ x_2, y_i\}) \leq \theta + 1
    $$
    \item [(2)] $h(x_i) > \theta$ for any $i \in \{1, 2\}$. Then, for any of the two bags in the second decomposition, say $\{x_1, y_1, \ldots, y_{\ell}\}$, we have 
    \begin{align*}
        h(\{x_1, y_1, \ldots, y_{\ell}\}) & \leq h(\{x_1, y_1\}) + h(\{x_1, y_2, \ldots, y_{\ell}\}) - h(\{x_1\}) \\
        & \leq h(\{x_1, y_1\}) + h(\{x_1, y_2\}) + h(\{x_1, y_3, \ldots, y_{\ell}\}) - 2 h(\{x_1\}) \\
        & \cdots  \\
        & \leq \sum_{i = 1}^{\ell - 1} h(\{x_1, y_i\}) - (\ell - 1) h(\{x_1\}) \\
        & \leq \ell - (\ell - 1) \theta
    \end{align*}
\end{enumerate}
Setting $\theta = 1 - 1/ \ell$, we conclude that $\subw(K_{2,\ell}) \leq 2 - 1/\ell$. Therefore, $\mathsf{emb}_{2\ell-1} = \subw =  2-1/\ell$ by \autoref{prop:lowerbound}.
\end{proof}

\begin{proof}[Proof of \autoref{tight:K33}]
   We first show that $\emb(K_{3, 3}) \geq 2$ by constructing an embedding $\psi$ from a $8$-clique to $K_{3, 3}$, where $\psi$ is defined as follows:
   \begin{align*}
       &\psi^{-1}(x_1) = \{1, 3, 5\}, \quad \psi^{-1}(x_2) = \{2, 4, 6\}, \quad \psi^{-1}(x_3) = \{7, 8\}  \\
       &\psi^{-1}(y_1) = \{1, 2\}, \quad \quad \psi^{-1}(y_2) = \{3 ,4\}, \quad \quad \psi^{-1}(y_3) = \{5, 6\} 
   \end{align*}
   Therefore,  $\emb \geq 8 / \mathsf{wed}(\psi) = 8 / 4 = 2$.
   
   Next, we show that $\subw(K_{3, 3}) \leq  2$. We take two tree decompositions, where the first decomposition has bags $\{x_1, x_2, x_3, y_i\}$, where $i \in [3]$ and the second decomposition has bags $\{y_1, y_2, y_3, x_i\}$, where $i \in [3]$. For an arbitrary edge-dominated $h \in \Gamma_3$, we assume WLOG in each decomposition, the maximum value of $h$ is attained when $i = 1$. We observe that
   \begin{align*}
       h(\{x_1, x_2, x_3, y_1\}) &\leq h(\{x_3, y_1\}) + h(\{x_1, x_2, y_1\}) - h(\{y_1\})  \leq  1 + h(\{x_1, x_2, y_1\}) - h(\{y_1\})   \\
       h(\{y_1, y_2, y_3, x_1\}) & \leq h(\{y_3, x_1\}) + h(\{y_1, y_2, x_1\}) - h(\{x_1\})  \leq  1 + h(\{y_1, y_2, x_1\}) - h(\{x_1\}) 
   \end{align*}
   Taking the sum of the above two inequalities, we get 
   \begin{align*}
       h(\{x_1, x_2, x_3, y_1\}) + h(\{y_1, y_2, y_3, x_1\}) & \leq 2 + h(\{x_1, x_2, y_1\}) - h(\{x_1\}) + h(\{y_1, y_2, x_1\}) - h(\{y_1\}) \\
       & \leq 2 + h(\{x_2, y_1\}) + h(\{y_2, x_1\}) \\
       & \leq 4
   \end{align*}
   Then, it holds that $\subw(K_{3, 3}) \leq \min \{h(\{x_1, x_2, x_3, y_1\}),  h(\{x_1, x_2, x_3, y_1\})\} \leq 2$. We close the proof by applying \autoref{prop:lowerbound}.
\end{proof}

\subsection{Chordal Queries} \label{app:tight:chordal}
% \begin{proof}[Proof of \autoref{lem:tds}]
% We first show that $(\htree, \chi)$ is a tree decomposition of $\mH$ if and only if it is also a tree decomposition of the clique-graph of $\mH$. The forward direction is straightforward. For the backward direction, let $(\htree, \chi)$ be a decomposition of the clique-graph of $\mH$. Then for any hyperedge $e \in \mH$ and any pair of vertices $u, v \in e$, we know that $\{t \mid u \in \chi(t)\} \cap \{t \mid v \in \chi(t)\} \neq \emptyset$. By the \textit{Helly Property}, there is a bag that contains all vertices in the hyperedge $e$. Therefore, $(\htree, \chi)$ is a tree decomposition for $\mH$. It is easy to extend the proof for proper tree decompositions.
% \end{proof}

\begin{proof}[Proof of \autoref{lem:properTD}]
Let $G$ be the clique-graph of $\mH$. Lemma 5.4 in~\cite{CarmeliKKK21} states that there is a one-to-one correspondence between the set of bags in a proper tree decomposition (of $G$) and the set of possible minimal triangulations of $G$. Thus, we proceed to prove the following claim: if a graph $G$ has only one minimal triangulation, then $G$ is chordal and so is $\mH$.

We show the contrapositive of the claim. Suppose the graph $G= (V, E)$ is not a chordal graph, then it has at least one minimal triangulation. We take one fill edge in this minimal triangulation, called $e$, and show that we can construct another minimal triangulation without taking $e$ as a fill edge. Indeed, we can start from a $|V|$-clique and remove $e$ from it. The resulting graph is called an almost-clique and shown to be chordal in~\autoref{app:tight}. Since $e \notin E$, this almost-clique is a triangulation of the original graph $G$ and we can keep removing fill edges from this triangulation till it becomes a minimal triangulation of $G$ without $e$. This process shows that there are at least two distinct minimal triangulations of $G$. This finishes the proof.
\end{proof}

\section{Missing Details from \autoref{sec:boat}} \label{app:gaps}
Before proving \autoref{prop:subwboat}, we show the following helper lemma.

\begin{lemma} \label{lem:4vertices}
Any proper tree decomposition of $Q_{hb}$ contains a bag that has at least $4$ vertices, two from $\{y_1, y_2, y_3\}$, two from $\{z_1, z_2, z_3\}$.
\end{lemma}
\begin{proof}
   Take any proper tree decomposition of $Q_{hb}$. By \autoref{lem:properTD}, it is also a proper tree decomposition of the clique-graph of $Q_{hb}$. Lemma 5.4 in~\cite{CarmeliKKK21} states that there is a one-to-one correspondence between the set of bags in a proper tree decomposition and the set of possible minimal triangulations of $G$ and furthermore, following \autoref{lem:maxclq}, the set of bags in the proper tree decomposition is exactly the set of maximal cliques after a minimal triangulation. Therefore, we only need to show the following statement: for any minimal triangulation of the clique-graph of $Q_{hb}$, there is a $4$-clique that contains at least $4$ vertices, two from $\{y_1, y_2, y_3\}$, two from $\{z_1, z_2, z_3\}$.

   Let us, WLOG, fill $\{y_2, z_3\}$ as a chord for the $4$-cycle $(y_2, y_3, z_3, z_2)$. For the $4$-cycle $(y_1, y_2, z_2, z_1)$, there are two cases:
   \begin{enumerate}
       \item fill $\{y_2, z_1\}$ as a chord for the $4$-cycle $(y_1, y_2, z_2, z_1)$: now for the $4$-cycle $(y_1, y_3, z_3, z_1)$, assume WLOG the chord $\{y_1, z_3\}$ is filled. This implied that $\{y_1, y_2, z_1, z_3\}$ forms a $4$-clique after this minimal triangluation.
       \item fill $\{y_1, z_2\}$ as a chord for the $4$-cycle $(y_1, y_2, z_2, z_1)$: in this case, consider the $4$-cycle $(y_1, y_3, z_3, z_2)$. Rose, Tarjan, and Lueker~\cite{RoseTL76} show that a triangulation is minimal if and only if every filled edge is the unique chord of a 4-cycle. This implies that only the chord $\{y_1, z_3\}$ can be filled now (not $\{y_3, z_2\}$), in order to be a minimal triangulation, which leads to a $4$-clique $\{y_1, y_2, z_2, z_3\}$.
   \end{enumerate}
   Since we have exhausted all possible minimal triangulations, the proof is finished. 
\end{proof}

Next, we show a formal proof of \autoref{prop:subwboat}.
\begin{proof}[Proof of \autoref{prop:subwboat}]
   To see that $\subw(Q_{hb}) \leq 2$, we note that there is fractional edge cover of $Q_{hb}$ that assigns weight $1$ and $\{y_1, y_2, y_3\}$ and $\{z_1, z_2, z_3\}$ and gets total weight of $2$. Thus, for $Q_{hb}$, $\subw \leq \fhw \leq 2$.

   Next, we show that $\subw(Q_{hb}) \geq 2$. We fix an edge-dominated submodular function defined on $V(Q_{hb})$  as follows,
   \begin{align*}
       & h(\emptyset) = 0, \quad  h(y_i) = h(z_i) = 1/2, \quad i \in [3] \\
       & h(e) = 1, \quad e \in E \\
       & h(\{y_1, y_2, y_3, z_i\}) = h(\{z_1, z_2, z_3, y_i\}) = 3/2, \quad i \in [3] \\
       & h(\{y_1, y_2, y_3, z_i, z_j\}) = h(\{z_1, z_2, z_3, y_i, y_j\}) = h(\{z_1, z_2, z_3, y_1, y_2, y_3\}) = 2, \quad i, j \in [3], i \neq j \\
       & h(\{u, v\}) = 1, \quad u, v \in V, u \neq v \\
       & h(\{y_i, y_j, z_k\}) = h(\{y_k, z_i, z_j\}) = 3/2, \quad i, j, k \in [3], i \neq j \\
       & h(\{y_i, y_j, z_k, z_{\ell}\}) = 2 , \quad i, j, k, \ell \in [3], i \neq j, k \neq \ell.
   \end{align*}
   By \autoref{lem:4vertices}, we know that for any proper decomposition, there is one bag $B$ such that $h(B) \geq h(\{y_i, y_j, z_k, z_{\ell}\}) = 2$, for some $i, j, k$, where $\ell \in [3], i \neq j, k \neq \ell$. Therefore, we have shown that $\subw(Q_{hb}) \geq 2$. Together, we have proved the claim. 
\end{proof}

\end{document}